\RequirePackage{etoolbox}
\csdef{input@path}{%
 {sty/}
 {img/}
}%

\documentclass[letterpaper]{imsart}
\firstpage{1}
\lastpage{1}

\usepackage{amsthm,amssymb}
\usepackage{amsmath,bm}
\usepackage{amsfonts}
\usepackage{mathrsfs}
\usepackage{multirow}
\usepackage{natbib}
\usepackage{bbm}
\usepackage{enumitem,stackengine}
\usepackage{afterpage}
\usepackage{csquotes}
\usepackage{appendix}
\usepackage{gensymb}
\usepackage{mathtools}
\usepackage{eufrak}
\usepackage{subfigure}
\newlength{\dhatheight}
\setlist{  
  listparindent=\parindent,
  parsep=0pt,
}

\newtheorem*{remark}{Remark}
\makeatletter
\usepackage[colorlinks,citecolor=blue,urlcolor=blue,filecolor=blue,backref=page]{hyperref}
\usepackage{graphicx}
\DeclareMathOperator*{\argmax}{arg\,max}
\DeclareMathOperator*{\argmin}{arg\,min}
\startlocaldefs
\numberwithin{equation}{section}
\theoremstyle{plain}
\newtheorem{thm}{Theorem}
\newtheorem{definition}{Definition}
\newtheorem{lemma}{Lemma}
\newcommand\barbelow[1]{\stackunder[1.2pt]{$#1$}{\rule{.8ex}{.075ex}}}
\endlocaldefs

\begin{document}


\begin{frontmatter}
\title{Bayesian Multivariate Quantile Regression Using Dependent Dirichlet Process Prior}

\runtitle{Bayesian Multiple-Output Quantile Regression}

\begin{aug}
\author{\fnms{Indrabati} \snm{Bhattacharya}\thanksref{addr1}
\ead[label=e1]{ibhatta@ncsu.edu}
}
\and
\author{\fnms{Subhashis} \snm{Ghosal}
\thanksref{addr1}
\ead[label=e2]{sghosal@ncsu.edu}}

\runauthor{Bhattacharya and Ghosal}

\address[addr1]{North Carolina State University\\
\printead{e1} 
    \printead*{e2}}


\end{aug}

\begin{abstract}
In this article, we consider a non-parametric Bayesian approach to multivariate quantile regression. The collection of related conditional distributions of a response vector $Y$ given a univariate covariate $X$ is modeled using a Dependent Dirichlet Process (DDP) prior. The DDP is used to introduce dependence across $x$. As the realizations from a Dirichlet process prior are almost surely discrete, we need to convolve it with a kernel. To model the error distribution as flexibly as possible, we use a countable mixture of multidimensional normal distributions as our kernel. For posterior computations, we use a truncated stick-breaking representation of the DDP. This approximation enables us to deal with only a finitely number of parameters. We use a Block Gibbs sampler for estimating the model parameters. We illustrate our method with simulation studies and real data applications. Finally, we provide a theoretical justification for the proposed method through posterior consistency. Our proposed procedure is new even when the response is univariate.
\end{abstract}


\end{frontmatter}




\section{Introduction}
Quantile regression is a popular alternative to the usual mean regression which models the relationship between the predictor and a specific quantile of the response. Univariate linear quantile regression was first proposed by \cite{koenker1978regression} and was extensively studied in the literature since then. Given a covariate $X \in \mathbb{R}^m$, the $\alpha$th linear quantile regression model for the response $Y \in \mathbb{R}$ can be written as $Q_{Y\vert x}(\alpha)=x^T\beta$, for $\alpha \in (0,1)$. Based on a sample $(X_1,Y_1),\dots, (X_n,Y_n)$, \cite{koenker1978regression} estimated the regression coefficient $\beta$ by the estimator
\begin{equation}\label{eq:1}
    \hat{\beta}=\argmin_{b \in \mathbb{R}^m}\sum_{i=1}^n\rho_\alpha(Y_i-X_i^Tb),
\end{equation}
with $\rho_\alpha(u)=u(\alpha-\mathbbm{1}(u<0))$, where $\mathbbm{1}$ is the indicator function. Fast algorithms to compute $\hat{\beta}$ were obtained in the literature; there is an R package called \textit{quantreg}.

\cite{yu2001bayesian} proposed a parametric Bayesian approach to quantile regression by assuming an asymmetric Laplace likelihood. As \cite{chang2015nonparametric} mentioned, modeling the error distribution directly by an asymmetric Laplace distribution is too restrictive. To avoid this restrictive parametric assumption, a number of non-parametric Bayesian approaches have been developed in the literature. A lot of these non-parametric methods are based on Dirichlet process mixture (DPM) models; see \cite{chang2015nonparametric} for a comprehensive review and references.

Univariate quantiles can be extended to the multivariate setting in a number of ways. There is no unique definition of quantiles in higher dimension because of the lack of a natural ordering of Euclidean space in higher dimension. \cite{chaudhuri1996geometric} introduced the notion of geometric quantiles which arises as a natural generalization of the spatial median (see \cite{small1990survey} for a review on spatial median). \cite{chakraborty2003multivariate} extended this idea to a regression framework. There are various other ways to define a multivariate quantile (\cite{serfling2002quantile}). \cite{hallin2010multivariate} introduced the notion of directional quantiles for multivariate location and multiple-output regression problems. The Bayesian literature on multivariate quantiles is very limited; only a few papers exist to the authors' knowledge. \cite{waldmann2015bayesian} considered bivariate quantile regression using a multivariate asymmetric Laplace likelihood; while \cite{drovandi2011likelihood} used a copula approach. Recently, \cite{guggisberg2019bayesian} proposed a Bayesian approach to the directional quantile framework developed in \cite{hallin2010multivariate}.

Our approach here is non-parametric, i.e., we model the collection of conditional distributions of the response $Y \in \mathbb{R}^k$ given a predictor $X \in \mathbb{R}$, without using any parametric family of distributions and then estimate the desired geometric quantile of the conditional distribution. The most commonly used prior for a probability distribution is the Dirichlet process prior. The collection of conditional distributions are viewed as related quantities and hence are modeled by a Dependent Dirichlet Process (DDP) (defined in Section \ref{2S2}). One major drawback of using a Dirichlet process prior is that almost all realizations from a Dirichlet process are discrete. This issue can be handled by convolving it with a kernel. To model the error distribution flexibly without any particular parametric form, we use a countable mixture of $k$-dimensional normal distributions as our kernel. We use a Block Gibbs sampler (see \cite{ishwaran2001gibbs}) for our posterior computations, which is considerably fast. The illustrations discussed here are focused on cases with bivariate response. It should be noted that our proposed method is a new contribution even in the context of univariate quantile regression.

We use geometric quantiles for our treatment of multivariate quantile regression here and hence we define it formally below. Consider the situation when the variable $Y$ is observed along with a univariate predictor $X$ lying in a compact interval $\mathfrak{X}$. For a given value $x$ of $X \in \mathfrak{X}$, let $P_{Y\vert x}$ stand for the conditional distribution of $Y$ given $X=x$, and $F_{Y\vert x}$ denote the CDF. Then the non-parametric multivariate quantile regression function of $Y$ is given by 
\begin{equation}
    Q_{Y\vert x}(u)= \argmin_{q \in \mathbb{R}^k}P_{Y\vert x}\{\Phi_2(u,Y-q)-\Phi_2(u,Y)\},
\end{equation}
with $\Phi_p(u,t)=\Vert t \Vert_p+ \langle u,t \rangle$, for $u\in B_q^{(k)}$, with $p^{-1}+q^{-1}=1$. The true conditional distribution of $Y$ given $x$ is denoted by $P_{Y\vert x}^\star$, with the CDF being $F_{Y\vert x}^\star$. The $u$th geometric quantile of the distribution $P_{Y\vert x}^\star$ is denoted by $Q_{Y\vert x}^\star(u)$. To estimate $Q_{Y\vert x}(u)$, it is sensible to assume that it changes gradually in $x$. Hence for sensible inference, we should pull information across neighboring values of $x$ by a smoothing technique. In a Bayesian setting that we follow, we achieve the objective by putting a suitable prior on the family of distributions $\{P_{Y \vert x}: x \in \mathfrak{X}\}$. The 
borrowing of information across neighboring values of $x$ may be introduced in the prior by using a Dependent Dirichlet Process (DDP) (see Section \ref{2S2}). However, as marginal distributions of a DDP are Dirichlet processes and hence $P_{Y \vert x}$ are discrete almost surely, a smoother version of the DDP by convolving with a kernel will be used as a prior. Then the resulting posterior distribution can be computed and the induced posterior distribution on the multivariate quantile regression can be used to obtain Bayes estimates and credible sets.

The rest of this paper is organized as follows. In Section \ref{2S2}, we give a brief background on Dependent Dirichlet Processes. In Section \ref{2S3}, we describe our non-parametric Bayesian modeling approach for multivariate quantile regression. Section \ref{2S4} gives the details of the posterior computations, and in Section \ref{2S5}, we give the posterior consistency theorems. In Sections \ref{2S6} and \ref{2S7}, we demonstrate the performance of our method on simulated and real data. We close the paper with the proofs in Section \ref{2S8}.

\section{Overview on Dependent Dirichlet Process}{\label{2S2}}
We use the notation $P \sim \mathrm{DP}(\alpha)$ to state that the random measure $P$ has a Dirichlet process distribution with base measure $\alpha$. We also use $P \sim \mathrm{DP}(MG)$ where $M=\lvert \alpha \rvert$ and $\bar{\alpha}=\alpha/M$ has a distribution function $G$.

Several papers have considered extending Dirichlet process models over related random distributions, for example, \cite{cifarelli1978nonparametric}, \cite{tomlinson1999analysis} and \cite{kottas2001bayesian}  (see \cite{de2004anova} for a comprehensive review), but these models were not naturally extended to include regression on covariates. Dependent Dirichlet Processes were introduced by \cite{maceachern1999dependent}, to address regression on a predictor variable. Following the notation in \cite{ghosal2017fundamentals}, suppose that, we have a collection of distributions $P_z$ on a sample space $\Omega$, indexed by a parameter $z$ belonging to some covariate space $\mathfrak{Z}$. A useful prior distribution on $P_z$ should treat them as related quantities.

It is reasonable to equip each marginal measure $P_z$ with a Dirichlet process prior. By the stick-breaking construction of Dirichlet process, we can write $P_z$ as
\begin{equation}\label{eq:5}
    P_z=\sum_{i=1}^{\infty}W_j(z)\delta_{\theta_j(z)},
\end{equation}
where $\{W_j(z):z \in \mathfrak{Z}\}$ are called the \enquote{stick-breaking weights} and $\{\theta_j(z):z \in \mathfrak{Z}\}$ are called \enquote{locations}. The stick breaking weights are constructed as
\begin{equation}
    W_j(z)=V_j(z)\prod_{l=1}^{j-1}(1-V_l(z)),
\end{equation}
where $V_j(z) \stackrel{iid}\sim \mathrm{Be}(1,M_z)$ for $M_z>0$ with $\mathrm{Be}(a,b)$ being a beta distribution with parameters $a$ and $b$. The locations $\theta_j(z)$ are i.i.d. draws from the base measure $G_z$. This representation ensures that $P_z \sim \mathrm{DP}(M_zG_z)$ for every $z \in \mathfrak{Z}$. A process satisfying this requirement is called a Dependent Dirichlet Process (DDP). For more details, see Section 14.9, \cite{ghosal2017fundamentals}.

Several variations of DDP models have been proposed over the years. \cite{de2004anova} described dependence across the random distributions in an ANOVA-type fashion; which is particularly useful for multivariate categorical covariates. \cite{nieto2012time} proposed a version of DDP that is suitable as a prior for a time series of random probability measures. \cite{gelfand2005bayesian} proposed a DDP for point-referenced spatial data, which was later extended by \cite{duan2007generalized}. \cite{sun2017location} proposed location dependent Dirichlet process (LDDP) that incorporates non-parametric Gaussian processes in the DP modeling framework to model dependency information among data arising from space and time.
\section{Bayesian Multivariate Quantile Regression with\\ DDP}
\label{2S3}
Consider a set of independent observations $(Y^n,X^n)=(Y_1,X_1),\dots,(Y_n,X_n)$ on a univariate predictor $X$ and a $k$-variate response $Y$. Let $\mathcal{F}=\{f:\mathbb{R}^k\times \mathfrak{X}\rightarrow [0,\infty],\linebreak \int_{\mathbb{R}^k}f(y\vert x)\mathrm{d}y=1\}$ be the space of conditional densities of $Y$ given $X$, where $\mathfrak{X} \subset \mathbb{R}$ is compact. Our goal is to infer about the collection $\{Q_{Y\vert x}(u): x \in \mathfrak{X}\}$, for some fixed $u \in B_2^{(k)}$. We adopt a non-parametric modeling framework in the following way,
\begin{equation}
    Y_i\vert x_i \sim f(\cdot\vert x_i),\ \{f(\cdot\vert x), x \in \mathfrak{X}\}\sim \Pi,
\end{equation}
where $\Pi$ denotes a prior for the class of conditional densities $\{f(\cdot \vert x): x \in \mathfrak{X}\}$. We choose a DDP prior here, but as we have mentioned, distributions drawn from a Dirichlet process prior are discrete. Hence, we have to use a kernel for convolving it with. This leads us to the model
\begin{equation}
    Y_i=\xi(X_i)+\varepsilon_i,\ i=1,\dots,n,
\end{equation}
with $\xi\sim \{G_x:x \in \mathfrak{X}\}$ independently, and $\varepsilon_i$ are the random errors. Since we are interested in the $u$th geometric quantile of $\{P_{Y\vert x}: x \in \mathfrak{X}\}$, we may choose the kernel in such a way that the $u$th geometric quantile of the error distribution is zero, that is, $Q_{\varepsilon}(u)=0$.

The collection of distributions $\{G_x:x \in \mathfrak{X}\}$ follows a DDP prior. We use a stick-breaking representation for $G_x$ to introduce dependence across $x$ as
\begin{equation}
  G_{x}=\sum_{l=1}^{\infty}W_l\delta_{\xi_l(x)},
\end{equation}
for any $x \in \mathfrak{X}$. We use a common set of stick-breaking weights $W_l$ which are constructed from $V_l \stackrel{iid}\sim \mathrm{Be}(1,M_1)$, for some $M_1>0$, where the locations are drawn from a $k$-dimensional Gaussian distribution with mean vector and covariance matrix varying with $x$. We represent the locations $\xi_l(x)$ as
\begin{equation}
    \xi_l(x)=\alpha_l+\beta_l(x),
\end{equation}
where $\alpha_l \stackrel{iid}\sim \mathrm{N}_k(c_0,\Sigma_0)$ and $\beta_l$ follows a Gaussian Process (GP) with mean function $c_1x$ and covariance kernel $\Sigma(x-x^\prime)$ where 
$
\Sigma (x)=\gamma \mathrm{diag} ( e^{-\lambda \vert x\vert},\ldots,e^{-\lambda \vert x\vert})
$. Here $c_0$ and $c_1$ are constant $k$-vectors and $\Sigma_0$ is a positive definite matrix of order $k\times k$. Also, $\gamma >0$ and $\lambda >0$ are constants. The justification behind choosing such a prior structure is that under this prior specification, $\mathrm{E}(\xi_l(x_i))=c_0+c_1x_i$, that is, under the prior, the expected value of the locations $\xi_l(x)$ varies linearly with $x$, and the sample paths vary smoothly.

We model the distributions of the observations by a countable location-scale mixture of $k$-dimensional Gaussian distributions through the relations 
\begin{equation}
    Y\vert \{\xi,\ X\} \sim \sum_{j=1}^\infty p_j\mathrm{N}_k(\xi(X)+\eta_j,\sigma_j^2I_k),
\end{equation}
where $I_k$ denotes the identity matrix of order $k$. This model allows for asymmetry in the error distribution. The weights $p_j,\ j=1,2,\dots$ are again constructed using stick-breaking, using $q_j \overset{iid}{\sim} \mathrm{Be}(1,M_2)$. We put independent $k$-dimensional Gaussian and inverse gamma hyper-priors on $\eta_j$ and $\sigma_j^2$ respectively. Note that the error distribution $\sum_{j=1}^\infty p_j\mathrm{N}_k(\eta_j, \sigma_j^2I_k)$ does not satisfy $Q_{\varepsilon}(u)=0$ for any $u \in B_2^{(k)}$. So for any $u \in B_2^{(k)}$, our quantile regression model is formulated as
\begin{equation}
    Q_{Y\vert x}(u)=\xi(x)+Q_{\varepsilon}(u).
\end{equation}
 To reduce the computational burden, we use a truncation approximation on both sets of stick-breaking weights. Thus the full hierarchical Bayesian model is given by
\begin{align*}
        &Y_i=\xi(X_i)+\varepsilon_i,\ i=1,\dots, n,\\
    &\varepsilon_i \overset{iid}{\sim}\sum_{j=1}^J p_j\mathrm{N}_k(\eta_j,\ \sigma_j^2I_k),\\ 
    & p_1=q_1,\ p_l=q_l \prod _{r=1}^{l-1}(1-q_r),\ l=2,\dots,J-1,\ p_J=1-\sum_{l=1}^{J-1}q_l,\\
    & \eta_j \overset{iid}{\sim}\mathrm{N}_k(c_{\eta},\ s_\eta^2I_k),\ \sigma_j^2\overset{iid}{\sim}\mathrm{IG}(a,b)\ \ j=1,2,\dots,J,\\
    & q_l \overset{iid}{\sim} \mathrm{Be}(1,M_2),\ l=1,2,\dots,J-1,\ M_2 \sim \mathrm{Ga}(a_{M_2},b_{M_2})\\
    &\xi(X_i)  \overset{ind}{\sim}G_{X_i},\ G_{X_i}=\sum_{l=1}^{N}W_l\delta_{\xi_l(X_i)}\\
    &\xi_l(X_i)=\alpha_l+\beta_l(X_i),\ l=1,\dots,N,\\
    &\alpha_1,\dots,\alpha_N \overset{iid}{\sim}\mathrm{N}_k(c_0,\ \Sigma_0),\\
    &(\beta_l(X_1),\dots,\beta_l(X_n))\overset{iid}{\sim}\mathrm{N}_{kn}((c_1X_1,\dots,c_1X_n)),\ (\!( \Sigma (X_i-X_j) )\!) ),\ j=1,\dots,N,\\
    & W_1=V_1,\ W_l=V_l\prod_{r=1}^{l-1}(1-V_r),\ l=2,\dots, N-1,\ W_N=1-\sum_{l=1}^{N-1} W_l,\\
    & V_l \overset{iid}{\sim} \mathrm{Be}(1,M_1),\ l=1,\dots,N-1,\ M_1 \sim \mathrm{Ga}(a_{M_1}, b_{M_1}),
\end{align*}
where $\mathrm{Ga}(a,b)$ denotes a Gamma distribution with parameters $a$ and $b$, and $\mathrm{IG}(a,b)$ denotes an inverse-gamma distribution. The truncated stick-breaking construction reduces the computations to finitely many terms. We use a Markov Chain Monte Carlo (MCMC) method to compute the posterior estimates, which we describe in detail in the next section.
\section{Posterior Computation}\label{2S4}

We use a block Gibbs sampler (\cite{ishwaran2001gibbs}) for estimating the parameters. Suppose that, we have $d$ distinct observations $X_1,\dots,X_d$ on the covariate $X$, and for each $X_i$, we have $n_i$ observations $Y_{i1},\dots,Y_{in_i}$ on the response $Y$. We introduce two sets of latent variables as below for the ease of computation.
\begin{itemize}
    \item Define $L=(L_1,\dots,L_d)$ such that $L_i=l$ if and only if $\xi(X_i)=\xi_l(X_i)$, $i=1,\dots,d$. 
    \item Also define  $Z=(Z_{11},\dots, Z_{dn_d})$ such that $\varepsilon_{im}\vert \{Z_{im}=j \}\sim \mathrm{N}_k(\eta_j,\sigma_j^2I_k)$.
\end{itemize}
We describe the posterior full conditional distributions for each parameter below.
\begin{enumerate}
\item To update $\alpha_1,\dots,\alpha_N$: 
    \begin{itemize}
    \item Let $d^{\star}$ be the number of distinct values $\{L_j^{\star}:j=1,\dots, d^{\star}\}$ of the vector $L$. If $l \notin \{L_j^{\star},:j=1,\dots,d^{\star}\}$, $\alpha_l$ is drawn from $\mathrm{N}_k(c_0,\Sigma_0)$. 
    \item If $l=L_j^{\star},\ j=1,\dots, d^{\star}$,
    \begin{equation*}
    \begin{split}
            p(\alpha_{L_j^{\star}}\vert-) \propto&\  \mathrm{exp}\{-\frac{1}{2}(\alpha_{L_j^{\star}}-c_0)^{\prime} \Sigma_0^{-1}(\alpha_{L_j^{\star}}-c_0)\}\\
            &\times\prod_{\{i:L_i=L_j^{\star}\}}\prod_{r=1}^{n_i}\mathrm{exp}\{-\frac{1}{2\sigma_{Z_{ir}}^2}(Y_{ir}-\alpha_{L_j^{\star}}-\beta_{L_j^{\star}}(X_i)-\eta_{Z_{ir}})^{\prime}\\
            &\qquad\qquad\qquad(Y_{ir}-\alpha_{L_j^{\star}}-\beta_{L_j^{\star}}(X_i)-\eta_{Z_{ir}})\},
            \end{split}
        \end{equation*}
        which is a $k$-dimensional Gaussian distribution with mean $$ \bigg\{\Sigma_0^{-1}+\\ \sum_{i: L_i=L_j^\star}\sum_{r=1}^{n_i}\frac{1}{\sigma_{Z_{ir}}^2}I_k\bigg\}^{-1}\bigg\{\Sigma_0^{-1}c_0+\sum_{i: L_i=L_j^\star}\sum_{r=1}^{n_i}\frac{(Y_{ir}-\alpha_{L_j^\star}-\beta_{L_j^\star}(x_i)-\eta_{Z_{ir}})}{\sigma_{Z_{ir}}^2}\bigg\}$$ and covariance matrix $\bigg\{\Sigma_0^{-1}+\sum_{i: L_i=L_j^\star}\sum_{r=1}^{n_i}\frac{1}{\sigma_{Z_{ir}}^2}I_k\bigg\}^{-1}$.
        \end{itemize}
        \item To update $\beta_l(X_1),\dots,\beta_l(X_d),\ l=1,\dots,N$:
        \begin{itemize}
        \item If $l \notin \{L_j^{\star},:j=1,\dots,d^{\star}\}$, $$(\beta_l(X_1),\dots, \beta_l(X_d))\sim\mathrm{N}_{kd}(c_1X_1,\dots,c_1X_d), (\!( \Sigma (X_i-X_j) )\!)).$$
        \item Denote $X^j=(X_i,\ i: L_i=L_j^{\star})$ and $X^{-j}=(X_i,\ i:L_i \neq L_j^{\star})$. Also, let $\beta_{L_j^{\star}}(X^{j})=(\beta_{L_j^{\star}}(X_i),\ i:L_i=L_j^{\star})$ and $\beta_{L_j^{\star}}(X^{-j})=(\beta_{L_j^{\star}}(X_i),\ i:L_i \neq L_j^{\star})$. Let $m_{j,-j}$ denote the conditional prior mean for $\beta_{L_j^{\star}}(X^j)$ given $\beta_{L_j^{\star}}(X^{-j})$. Also, let $V_{j,-j}$ denote the conditional prior covariance matrix for $\beta_{L_j^{\star}}(X^{j})$ given $\beta_{L_j^{\star}}(X^{-j})$. Similarly, $m_{-j,j}$ and $V_{-j,j}$ denote the conditional prior mean and covariance matrix of $\beta_{L_j^{\star}}(X^{-j})$ given $\beta_{L_j^{\star}}(X^{j})$ respectively. If $l=L_j^{\star}$, $j=1,\dots,d^{\star}$, then for $i$ such that $L_i=L_j^\star$
        \begin{equation*}
           \begin{aligned} p(\beta_{L_j^{\star}}(X^{j})\vert-)\propto\ \mathrm{exp}&\{-\frac{1}{2}(\beta_{L_j^{\star}}(X^{j})-m_{j,-j})^{\prime} V_{j,-j}^{-1}(\beta_{L_j^{\star}}(X^{-j})-m_{j,-j})\}\times\\
           &\prod_{\{i:L_i=L_j^{\star}\}}\prod_{r=1}^{n_i}\mathrm{exp}\{-\frac{1}{2\sigma_{Z_{ir}}^2}(Y_{ir}-\alpha_{L_j^{\star}}-\beta_{L_j^{\star}}(X_i)-\eta_{Z_{ir}})^\prime\\
           &\qquad \qquad\qquad\qquad(Y_{ir}-\alpha_{L_j^{\star}}-\beta_{L_j^{\star}}(X_i)-\eta_{Z_{ir}})\},
           \end{aligned}
           \end{equation*}
           which is a $kU_{L_j^\star}$-dimensional Gaussian distribution, where $U_l=\#\{i: L_i=l\}$.
        \item Also $p(\beta_{L_j^{\star}}(X^{-j})\vert-)$ is proportional to
        \begin{equation*}
           \begin{aligned}  \mathrm{exp}\{-\frac{1}{2}(\beta_{L_j^{\star}}(X^{-j})-m_{-j,j})^{\prime} V_{-j,j}^{-1}(\beta_{L_j^{\star}}(X^{-j})-
           m_{-j,j})\},
           \end{aligned}
           \end{equation*}
           which is a $k(d-U_{L_j^\star})$-dimensional Gaussian distribution with mean vector $m_{-j,j}$ and covariance matrix $V_{-j,j}$.
           \item If our goal to estimate $Q_{Y\vert x}(u)$ for an arbitrary $x \in \mathfrak{X}$, for a fixed $u \in B_2^{(k)}$, $\beta_l(x)$ has to be sampled conditional on $\beta_l(X_1),\dots,\beta_l(X_d)$, for $l=1,\dots,N$, which is a $k$-dimensional Gaussian distribution.
        \end{itemize}
        \item To update $W_1,\dots,W_N$:
        
        The posterior full conditional for $W$ is given by
        \begin{equation*}
            p(w\vert -) \propto f_W(w\vert M_1)\prod_{l=1}^Nw_l^{U_l},
        \end{equation*}
        where $f_W(\cdot\vert M_1)$ is the generalized Dirichlet distribution (\cite{wong1998generalized})
        \begin{equation*}
     f_W(w\vert M_1)=M_1^{N-1}  w_N^{M_1-1} (1-w_1)^{-1} (1-(w_1+w_2))^{-1}\times \dots \times   (1-\sum_{l=2}^{N-2}w_l)^{-1}.
        \end{equation*}
        The posterior for $W$ is a generalized Dirichlet distribution as well and can be sampled as follows using latent Beta variables as follows.
        \begin{itemize}
            \item Generate $V_l \stackrel{ind}\sim \mathrm{Be}(1+U_l,M_1+\sum_{r=l+1}^NU_r)$ for $l=2,\dots,N-1$.
            \item Set $W_1=V_1$, $W_l=V_l\prod_{r=1}^{l-1}(1-V_r)$, for $l=2,\dots,N-1$, and $W_N=1-\sum_{l=1}^{N-1}W_l$.
        \end{itemize}
        \item To update $L_1,\dots,L_d$:\par
        Each $L_i$ is drawn from $\{1,\dots,N\}$ with probabilities proportional to
        \begin{equation*}
            W_l\prod_{r=1}^{n_i}
            \mathrm{exp}\{-\frac{1}{2\sigma_j^2}(Y_{ir}-\alpha_l-\beta_l(X_i)-\eta_{Z_{ir}})^{\prime}(Y_{ir}-\alpha_l-\beta_l(X_i)-\eta_{Z_{ir}})\},
        \end{equation*}
        for $l=1,\dots,N$.
        \item To update $\eta_1,\dots,\eta_J$:\par
        Suppose there are $s$ distinct values of $Z_{ij},i=1,\dots,d, j=1,\dots,n_i$, and we denote them by $Z^{\star}_1,\dots,Z_s^{\star}$. 
        
        \begin{itemize}
            \item If $j \notin \{Z^{\star}_r:r=1,\dots,s\}$, draw $\eta_j$ from $\mathrm{N}_k(c_{\eta},s_{\eta}^2I_k)$.
            \item If $j \in \{Z^{\star}_r:r=1,\dots,s\}$, draw $\eta_j$ from
                    \begin{equation*}
                    \begin{split}
        p(\eta_{Z^{\star}_r}\vert-) \propto\ & \mathrm{exp}\{-\frac{1}{2s_{\eta}^2}(\eta_{Z_r^{\star}}-c_{\eta})^{\prime} (\eta_{Z_r^{\star}}-c_{\eta})\}\\
            &\prod_{\{(i,l):Z_{il}=Z_r^{\star}\}}\mathrm{exp}\{-\frac{1}{2\sigma_{Z_r^{\star}}^2}(Y_{il}-\alpha_{L_i}-\beta_{L_i}(X_i)-\eta_{Z_r^{\star}})^{\prime}\\
        &\qquad\qquad\qquad(Y_{il}-\alpha_{L_i}- \beta_{L_i}(X_i)-\eta_{Z_r^{\star}})\},
            \end{split}
            \end{equation*}
            which is a $k$-dimensional Gaussian distribution, with mean vector $$ \bigg(\frac{1}{s_{\eta}^2}+\frac{1}{\sigma_{Z_r^\star}^2}\bigg)^{-1}\bigg(\frac{1}{s_\eta^2}c_\eta+\sum_{i,l: Z_{il}=Z_r^\star}\frac{Y_{il}-\alpha_{L_i}-\beta_{L_i}(X_i)-\eta_{Z_r^\star}}{\sigma_{Z_r^\star}^2}\bigg)$$ and covariance matrix $\bigg(\frac{1}{s_{\eta}^2}+\frac{1}{\sigma_{Z_r^\star}^2}\bigg)^{-1}I_k$.
        \end{itemize}
        \item To update $\sigma_1^2,\dots,\sigma_J^2$:
        \begin{itemize}
            \item If $j \notin \{Z^{\star}_r:r=1,\dots,s\}$, draw $\sigma_j^2$ from $\mathrm{IG}(a,b)$.
            \item If $j \in \{Z^{\star}_r:r=1,\dots,s\}$, $p(\sigma_{Z^{\star}_r}^2\vert-)$ is proportional to
            \begin{equation*}
              \begin{split}  (\sigma_{Z_r^{\star}}^2)^{-a-1}e^{-b/\sigma_{Z_r^{\star}}^2}\prod_{\{(i,l):Z_{il}=Z_r^{\star}\}}&\mathrm{exp}\{-\frac{1}{2\sigma_{Z_r^{\star}}^2}(Y_{il}-\alpha_{L_i}-\beta_{L_i}(X_i)-\eta_{Z_r^{\star}})^{\prime}\\
              &(Y_{il}-\alpha_{L_i}-\beta_{L_i}(X_i)-\eta_{Z_r^{\star}})\}.
               \end{split}
            \end{equation*}
            
        \end{itemize}
        \item To update $Z_{il},i=1,\dots,d,l=1,\dots,n_i$: 
        \begin{itemize}
        \item Draw $Z_{il}$ from $\{1,\dots, J\}$ with probability proportional to
            \begin{equation*}
                q_j\prod_{\{(i,l):Z_{il}=j\}}\mathrm{exp}\{-\frac{1}{2\sigma_j^2}(Y_{il}-\alpha_{L_i}-\beta_{L_i}(X_i)-\eta_j)^{\prime}(Y_{il}-\alpha_{L_i}-\beta_{L_i}(X_i)-\eta_j)\}.
            \end{equation*}
            \end{itemize}
            \item To update $p_1,\dots,p_J$:

             The posterior full conditional for $p=(p_1,\dots,p_J)$ is given by
        \begin{equation*}
            p(p_1,\dots,p_J\vert -) \propto f_p(p_1,\dots,p_J\vert M_2)\prod_{j=1}^J p_j^{U_j^\star},
        \end{equation*}
        with $U_j^\star=\#\{(i,l): Z_{il}=j\}$, and $f_p(p_1,\dots,p_J\vert M_2)$ is the generalized Dirichlet distribution
        \begin{equation*}
     f_p(p_1,\dots,p_J\vert M_2)=M_2^{J-1}  p_J^{M_2-1} (1-p_1)^{-1} (1-(p_1+p_2))^{-1}\times \dots \times   (1-\sum_{l=2}^{J-2}p_l)^{-1}.
        \end{equation*}
        The posterior for $(p_1,\dots,p_J)$ is a generalized Dirichlet distribution as well and can be sampled as follows using latent Beta variables as follows.
        \begin{itemize}
            \item Generate $q_l \stackrel{ind}\sim \mathrm{Be}(1+U_l^\star,M_2+\sum_{r=l+1}^J U_r^\star)$ for $l=2,\dots,J-1$.
            \item Set $p_1=q_l$, $p_l=q_l\prod_{r=1}^{l-1}(1-q_l)$, for $l=2,\dots,J-1$, and $W_J=1-\sum_{l=1}^{J-1}W_l$.
        \end{itemize}
        \item 
       To update $W_1,\dots,W_N:$
       
        The posterior for $(W_1,\dots,W_N)$ is also generalized Dirichlet distribution and can be sampled as follows using latent Beta variables as follows.
        \begin{itemize}
            \item Generate $V_l \stackrel{ind}\sim \mathrm{Be}(1+U_l,M_2+\sum_{r=l+1}^NU_r)$ for $l=2,\dots,N-1$.
            \item Set $W_1=V_1$, $W_l=V_l\prod_{r=1}^{l-1}(1-V_r)$, for $l=2,\dots,N-1$, and $W_N=1-\sum_{l=1}^{N-1}W_l$.
        \end{itemize}
        \item To update $M_1$:
            \begin{itemize}
                \item The posterior full conditional for $M_1$ is given by
                \begin{equation*}
                    p(M_1\vert -) \propto e^{-(b_{M_1}-\log W_N)M_1}M_1^{a_{M_1}+N-1},
                \end{equation*}
                which is a $\mathrm{Ga}(a_{M_1}+N,b_{M_1}-\log W_N)$ distribution.
            \end{itemize}
            \item To update $M_2$:
            \begin{itemize}
            \item The posterior full conditional for $M_2$ is given by
                \begin{equation*}
                    p(M_2\vert -) \propto e^{-(b_{M_2}-\log p_J)M_2}M_2^{a_{M_2}+J-1},
                \end{equation*}
                which is a $\mathrm{Ga}(a_{M_2}+J,b_{M_2}-\log p_J)$ distribution.
                \end{itemize}
\end{enumerate}
The MCMC sampling scheme is implemented using the \textit{nimble} package in R, which is a system for writing hierarchical Bayesian models highly compatible with \textit{BUGS} and \textit{JAGS}. The advantage of using \textit{nimble} is that it compiles the models by generating C++ code, which makes the computation faster.

After generating the posterior samples, we need to take care of the violation of the condition $Q_{\varepsilon}(u)=0$. For each of the $B$ many MCMC iterations of $p=(p_1,\dots,p_J),\ \eta=(\eta_1,\dots,\eta_J)$ and $\sigma^2=(\sigma_1^2,\dots,\sigma_J^2)$, we have to compute
\begin{equation}\label{eq:6}
    Q^b_{\varepsilon}(u)=\argmin_{\theta \in \mathbb{R}^k}\sum_{j=1}^Jp_j^b\int \Phi_2(u,x-\theta)e^{-\Vert x-\eta_j^b \Vert ^2/2(\sigma_j^2)^b}\mathrm{d}x,\ b=1,\dots, B.
\end{equation}
When the response variable is two-dimensional, the integral inside \eqref{eq:6} can be reduced to a one-dimensional integral by the following trick. For $k=2$, we can do a polar transform and reduce the integral for each $j$ in \eqref{eq:6} to 
\begin{equation}\label{eq:7}
\begin{aligned}
  & \bigg(\frac{1}{2\pi}\bigg)^{k/2}e^{-\Vert \theta-\eta_j \Vert^2/2 \sigma_j^2} \sigma_j^2  \int_0^{\infty} r^2e^{-r^2/2}\bigg(\int_0^{2\pi}e^{-r(\theta_1-\eta_{j1})\cos u-r(\theta_2-\eta_{j2})\sin u}\mathrm{d}u\bigg)\mathrm{d}r\\
   & +u_1(\eta_{j1}-\theta_1)+u_2(\eta_{j2}-\theta_2).
   \end{aligned}
\end{equation}
Then, \eqref{eq:7} reduces to
\begin{equation}\label{eq:8}
\begin{aligned}
    &\bigg(\frac{1}{2\pi}\bigg)^{k/2}e^{-\Vert \theta-\eta \Vert^2/2\sigma_j^2}\sigma_j^2  \int_0^{\infty} r^2e^{-r^2/2}2\pi I_0\bigg(r\sqrt{(\theta_1-\eta_1)^2+(\theta_2-\eta_2)^2}\bigg)\mathrm{d}r\\
    & u_1(\eta_1-\theta_1)+u_2(\eta_2-\theta_2),
    \end{aligned}
\end{equation}
where $I_0$ is the modified Bessel function of first kind (\cite{ifantis1990inequalities}). We can take a rectangular grid for $\theta$ and for each $\theta$ in that grid, we compute the integral in \eqref{eq:8} and the minimizer gives us an approximation for $Q_{\varepsilon}^b(u)$. We then calculate $B^{-1}\sum_{b=1}^BQ_{\varepsilon}^b(u)$ and add this to the posterior mean of $\xi(X_i)$ for each $i=1,\dots,d$, which gives us the $u$th geometric quantile for $Y$ given $X_i$, $i=1,\dots,d$.

There is another way of numerically computing $Q_{\epsilon}^b(u)$, by Monte Carlo integration. This method extends to dimensions higher than 2. For each $b \in \{1,\dots,B\}$, we generate a large number of samples $t_1^b,\dots, t_R^b$ from the mixture of $k$-dimensional Gaussian distributions $\sum_{j=1}^Jp_j\mathrm{N}_k(\eta_j,\sigma_j^2I_k)$. We again take a grid for $\theta$ and for each $\theta$ in the grid and for each $b \in \{1,\dots,B\}$. We compute the Monte Carlo average $R^{-1}\sum_{r=1}^R\Phi_2(u,t_r^b-\theta)$. We look at the minimizer of this Monte Carlo  average which approximates $Q_{\varepsilon}^b(u)$.
\section{Posterior Consistency}
\label{2S5}
In this section, we prove the weak posterior consistency for any fixed $u$th geometric quantile of the collection of true conditional densities $\{f^\star(\cdot \vert x): x \in \mathfrak{X}\}$. The model described in Section \ref{2S3} can be alternatively written as
\begin{equation}
    f(y\vert x)= \int \frac{1}{\sigma^k}\phi_k \bigg(\frac{y-\xi(x)-\eta}{\sigma}\bigg)\mathrm{d}G(\xi)\mathrm{d}Q(\eta,\sigma),
\end{equation}
where $\phi_k(\cdot)$ denotes a $k$-dimesnional standard normal kernel, and $G \times Q$ is the mixing distribution. The measures $G$ and $Q$ are of the form
\begin{align*}
    G=& \sum_{h=1}^\infty W_h \delta_{\xi_h},\\
    Q=& \sum_{j=1}^\infty p_j \delta_{(\eta_j,\sigma_j)}.
\end{align*}
For every $x \in \mathfrak{X}=[0,1]$, $G_x$ is the induced measure of $G$ through the evaluation map $\xi \mapsto \xi(x)$, which can be written as
$$
G_x=\sum_{h=1}^\infty W_h\delta_{\xi_h(x)}.
$$
The mixing distribution $G\times Q$ is given the prior $\mathcal{P}$ on $\mathfrak{M}(\mathcal{C}(\mathfrak{X})\times \mathbb{R}^k \times \mathbb{R}^+)$, where $\mathfrak{M}(\Theta)$ is the space of probability distributions on $\Theta$, and $\mathcal{C}(\mathfrak{X})$ is the space of continuous functions on $\mathfrak{X}$. The prior $\mathcal{P}$ on $\mathfrak{M}(\mathcal{C}(\mathfrak{X})\times \mathbb{R}^k \times \mathbb{R}^+)$ induces a prior $\Pi$ on the space of conditional densities $\mathcal{F}$ through the map $ \displaystyle G \times Q \mapsto \int \phi_k \bigg(\frac{y-\xi(x)-\eta}{\sigma}\bigg)\mathrm{d}G(\xi)\mathrm{d}Q(\eta,\sigma)$. We put priors on $G$ and $Q$ through stick-breaking weights, as described in Section \ref{2S3}. The true density $f^\star \in \mathcal{F}$ is assumed to be of the form
\begin{equation}{\label{eq72}}
    f^\star(y\vert x)=\int \frac{1}{\sigma^k}\phi_k \bigg(\frac{y-\xi(x)-\eta}{\sigma}\bigg)\mathrm{d}G^\star(\xi)\mathrm{d}Q^\star(\eta,\sigma),
\end{equation}
where $G^\star$ and $Q^\star$ are compactly supported probability measures on $\mathcal{C}(\mathfrak{X})$ and $\mathbb{R}^k\times \mathbb{R}^+$ respectively. Just like before, $G_x^\star$ is the induced measure from $G^\star$ through the evaluation map $\xi \mapsto \xi(x)$, with $G^\star$ being compactly supported on $\mathbb{R}^k$.
To prove the posterior consistency of the conditional geometric quantiles $Q_{Y\vert x}(u)$, we need to show the posterior consistency of the conditional distribution $F_{Y\vert x}$ around a neighborhood of the true distribution $F^\star_{Y\vert x}$. It appears not possible to derive the posterior consistency of the conditional distribution at a value of $x$ from the weak posterior consistency of the joint distribution of $X$ and $Y$. However, we can derive the posterior consistency of a $\delta$-smoothed conditional distribution of $Y$ given $x$, defined below. Let $P_{X,Y}$ denote the joint distribution of $X$ and $Y$, and let $F_{X,Y}$ denote the CDF. for a chosen $\delta >0$, the $\delta$-smoothed posterior distribution function of $Y$ given $X$ is defined as
\begin{equation}
    F_{\delta; Y\vert x}(y)= \frac{P_{X,Y}(\vert X-x \vert \leq \delta, Y \leq y) }{P_X(\vert X-x \vert \leq \delta)}.
\end{equation}
For $u \in B_2^{(k)}$, the $u$th geometric quantile of the $\delta$-smoothed conditional distribution is defined as
\begin{equation}
Q_{\delta;Y\vert x}(u)=   \argmin_{\theta \in \mathbb{R}^k}\int \{\Phi_2(u,Y-\theta)-\Phi_2(u,Y)\}\mathrm{d}F_{\delta; Y \vert x}(y).
\end{equation}
Our main theorem, which gives the posterior consistency for $\{Q_{\delta;Y\vert x}(u): x\in \mathfrak{X}\}$ for a fixed $u \in B_2^{(k)}$ is stated below. 
\begin{thm}\label{th1}
Assume that, for every $x \in \mathfrak{X}$ and $u \in B_2^{(k)}$, 
$$
\inf_{\theta: \Vert \theta-\theta^\star\Vert_2 \geq \epsilon} \int \{\Phi_2(u,y-\theta)-\Phi_2(u,y)\}\mathrm{d}F^\star_{Y\vert x}(y) > \int \{\Phi_2(u,y-\theta^\star)-\Phi_2(u,y)\}\mathrm{d}F^\star_{Y\vert x}(y).
$$
Then for $\delta_n\to 0$ sufficiently slowly and for every $x \in \mathfrak{X}$, 
$\Pi\{\vert Q_{\delta_n;Y\vert x}(u)- Q^\star_{Y\vert x}(u)\vert <\epsilon \vert (Y^n,X^n)\} \rightarrow 1$ a.s., for every $\epsilon>0$.
\end{thm}
We first need to introduce some notions and establish some auxiliary results. The first result guarantees that our chosen prior is sensible in the non-parametric setting, i.e., it has a large topological support. We consider support properties for weak neighborhoods of the type 
\begin{equation}\label{eq455}
\begin{split}
    \bigg\{f: \sup_x\Big\vert \int\{g(y)f(y\vert x)\mathrm{d}y-g(y)f^\star(y\vert x)\}\mathrm{d}y\Big\vert < \epsilon\bigg\},
\end{split}
\end{equation}
for every bounded and continuous function $g: \mathbb{R}^k\to \mathbb{R}$. 
\begin{lemma}{\label{l451}}
For every bounded and continuous $g:\mathbb{R}^k \rightarrow [0,1]$, 
\begin{equation}
\begin{split}
    \Pi\bigg\{f: \sup_x\Big\vert \int\{g(y)f(y\vert x)\mathrm{d}y-g(y)f^\star(y\vert x)\}\mathrm{d}y\Big\vert < \epsilon\bigg\}>0.
\end{split}
\end{equation}
\end{lemma}
The proof of Lemma \label{l451} in Section \ref{2S8}. We will need one more auxiliary result for proving the weak consistency, and for that we will need a bit more notation here. We define the joint density of $X$ and $Y$ by $h(x,y)=f(y\vert x)q(x)$, where $q(x)$ is the fixed marginal density of $X$, assumed to be bounded away from 0 on $\mathfrak{X}$. Similarly, the true joint density of $X$ and $Y$ is given by $h^\star(x,y)=f^\star(y\vert x)q(x)$. 
Next, we define the notion of a weak neighborhood for the true conditional density $f^\star \in \mathcal{F}$. 
\begin{definition}
A sub-base of a weak-neighborhood for a collection of conditional densities $\{f^\star(\cdot \vert x): x \in \mathfrak{X}\}$ is defined as
\begin{equation}{\label{eq57}}
    W_{\epsilon,g}(f^\star)=\bigg\{f: \big\vert \int_{\mathfrak{X}\times \mathbb{R}^k} gh-\int_{\mathfrak{X}\times \mathbb{R}^k} gh^\star \big\vert < \epsilon \bigg\},
\end{equation}
for a bounded and continuous function $g:\mathfrak{X}\times \mathbb{R}^k$. A weak neighborhood base is formed by finite intersection of neighborhoods of the type \eqref{eq57}.
\end{definition}
\begin{definition}
The posterior $\Pi\{\cdot \vert (Y^n,X^n)\}$ is weakly consistent at $f^\star \in \mathcal{F}$ if for any bounded and continuous function $g$, $\Pi\{W_{\epsilon,g}(f^\star)\vert (Y^n,X^n)\}\rightarrow 1$ a.s.
\end{definition}
However, the large topological support of the prior is not sufficient for proving the weak posterior consistency, The weak posterior consistency at $h^\star$ holds if the prior $\Pi$ puts positive probability on Kullback-Leibler neighborhoods of $f^\star$, which is defined below.
\begin{definition}
For any $\epsilon > 0$, an $\epsilon$-sized Kullback-Leibler ($\mathrm{KL}$) neighborhood around $f^\star$ is defined as
$$
K_{\epsilon}(f^\star)=\{f: \mathrm{KL}(h^\star,h)< \epsilon,\ h(x,y)=f(y\vert x)q(x), x \in \mathfrak{X}, y \in \mathbb{R}^k\},
$$
where $\mathrm{KL}(h^\star,h)=\int h^\star \log (h^\star/h)$. Then if $\Pi\{K_{\epsilon}(f^\star)\}>0$ for every $\epsilon > 0$, we say $f^\star \in \mathrm{KL}(\Pi)$. 
\end{definition}
\begin{lemma}{\label{l1}}
For $f^\star \in \mathcal{F}$ of the form in \eqref{eq72}, $f^\star \in \mathrm{KL}(\Pi)$. 
\end{lemma}
The proof of Lemma \ref{l1} is presented in Section \ref{2S8}, which ensures the weak consistency of the posterior at $h^\star$.
\section{Simulation Study}\label{2S6}
In this section, we demonstrate Bayesian median regression with a bivariate response $Y$ and a univariate predictor $X$ on simulated data, which is a special case of the method illustrated above. We compare our method with a couple of frequentist methods. 

It would be interesting to investigate the cases where the error distribution is something other than Gaussian, so here we choose a bivariate $t$-distribution with degree of freedom 1, non-centrality parameter $(0,0)$ and scale matrix $I_2$ (a symmetric heavy-tailed distribution) and a bivariate gamma distribution with shape and rate parameter 1 and correlation matrix
$$
V=
\begin{pmatrix}
1 & 0.7\\
0.7 & 1
\end{pmatrix},
$$
which is a skewed distribution. We draw $100$ samples $\varepsilon=(\varepsilon_1,\dots,\varepsilon_{100})$ from the above mentioned error distributions, and the predictors $X=(X_1,\dots,X_{100})$ are drawn from a $\mathrm{N}(0,1)$ distribution. We form the response vector $Y=(Y_1,Y_2)$ as follows.
\begin{equation}
    \begin{pmatrix}
    Y_1\\
    Y_2
    \end{pmatrix}
    = \begin{pmatrix}
    1 & 2 \\
    0 & 1
    \end{pmatrix}
    \begin{pmatrix}
    1\\
    X^2
    \end{pmatrix} + \epsilon.
\end{equation}
We are interested in the $(0,0)$-th geometric quantile, that is, the spatial median of $y$ given $x$. Next, we describe our chosen prior specifications. We have two sets of stick-breaking weights, $(p_1,p_2,\dots,)$ and $(W_1,W_2,\dots,)$. Both sets of stick-breaking weights are truncated at 20, i.e., we have chosen both $N$ and $J$ to be 20. Both set of stick-breaking weights are generated from the variables $V_l \overset{iid}{\sim} \mathrm{Be}(1,M_1)$ and $q_j \overset{iid}{\sim} \mathrm{Be}(1,M_2)$, $l=1,\dots,20,\ j=1,\dots,20$, where $M_1$ and $M_2$ are drawn from $\mathrm{Gamma}(1,1)$ prior. Next, $(\eta_1,\dots,\eta_{20})$ is drawn from a $N_2((0,0)^T,10I_2)$ prior, and $(\sigma_1^2,\dots,\sigma_{20}^2)$ is drawn from a $\mathrm{IG}(1,1)$ prior. We choose $c_0$ to be $(1,1)^T$ and $c_1$ to be $(2,1/2)$. Also, $S_0$ is chosen to be $10I_2$. For the matrix $(\!( \Sigma (X_i-X_j) )\!)$, $\gamma$ is chosen to be $10$ and $\lambda$ is chosen to be $1/2$.

Since we could only prove a weak consistency theorem for quantiles of $\delta$-smoothed posterior distributions, we here demonstrate quantile regression for $\delta$-smoothed posterior distribution of $Y$ given $X$, for some chosen $\delta$. The weak posterior consistency ensures that, for some $\epsilon>0$
\begin{equation}
   \Pi\{ \sup_{x,y}\vert F_{X,Y}(x,y)-F_{X,Y}^\star(x,y)\vert<\epsilon \vert (Y^n,X^n)\}\rightarrow 1,
\end{equation}
It will be shown in Section \ref{2S8} that $\delta$ should be bigger than $\epsilon$ for Theorem \ref{th1} to hold. We did not prove a convergence rate theorem here, but intuitively the rate of convergence should be $n^{-1/2}$. Hence we choose $\delta=\delta_n$ to be bigger then $n^{-1/2}$, say $n^{-1/3}$.

To compute the spatial median for the $\delta$-smoothed posterior for the conditional distribution of $Y$ given $X$, for each value $x$ of $X$, we draw a sample $\Tilde{x}$ from $\mathrm{N}(0,1)$ truncated in $[x-\delta,x+\delta]$, and we draw samples from the posterior distribution of $Y$ given $\Tilde{x}$ following the steps in Section \ref{2S4}. Note that the spatial median of the true error distribution is equal to its center of symmetry, $0$. We generate 5000 samples from the posterior distribution with a burn-in of 500. We report a mean square error (MSE) for the conditional spatial median $\xi$, which is given by
\begin{equation*}
    \text{MSE}=\frac{1}{100} \sum_{i=1}^{100}\Vert \xi^\star(X_i)-\bar{\xi}(X_i)\Vert_2^2,
\end{equation*}
where $\bar{\xi}(x)$ is the posterior mean of the $\delta$-smoothed spatial median of $Y$ given $x$ and $\xi^\star(x)$ is the true conditional spatial median of $y$ given $x$. One well-known competing method is the frequentist linear multivariate median regression proposed by \cite{bai1990asymptotic}, where the regression estimates are obtained by minimizing $\sum_{i=1}^n \Vert Y_i-\alpha-\beta X_i \Vert_2$ with respect to $\alpha$ and $\beta$. Another method is a non-parametric version of bivariate median regression method. For every $x$, we estimate the conditional median $\xi(x)$ by minimizing $\argmin_{\theta \in \mathbb{R}^k}\sum_{i=1}^{100}\Vert y_i-\theta \Vert_2 p_i(x)$, where $$p_j(x)=\frac{K((x-X_j)/h)}{\sum_{i=1}^{100}K((x-X_i)/h)}.$$
The bandwith $h$ has been chosen using cross-validation, and $K$ here is a standard normal kernel. The mean square errors for all three methods are shown in Table \ref{tab1}. Our method gives a lower MSE than the other two methods, thus achieving a gain.
\begin{table}[ht]
\centering
\begin{tabular}{rrrr}
  \hline
Error distribution & NP-Bayes & Frequentist linear &  NP-frequentist\\ 
  \hline
Bivariate $t$ & 9.40 & 17.43 & 14.52\\
Bivariate gamma & 7.29 & 18.02 & 14.98\\
   \hline
\end{tabular}
\caption{MSE's for conditional spatial medians for our method (NP-Bayes), the frequentist linear median regression and the non-parametric median regression (NP-frequentist).}
\label{tab1}
\end{table}
\section{Application to Blood Pressure data}\label{2S7}
We do Bayesian bivariate median regression using our method on a data that appeared in \cite{chakraborty1999multivariate}. This data (shown in Table \ref{tab:Table 2}) was collected by the Biological Sciences Division of the Indian Statistical Institute, Kolkata. This data contains systolic and diastolic blood pressures of 40 Marwari (an Indian ethnic group) females living in the Burrabazar area of Kolkata. 
\begin{table}[ht]
\centering
\begin{tabular}{rrrrrrrr}
  \hline
Serial & Age & Systolic & Diastolic & Serial & Age & Systolic & Diastolic\\ No. & & Pressure & Pressure & No. & & Pressure & Pressure\\ 
  \hline
  2 &  21 & 120 &  88 &  22 &  76 & 160 &  90 \\ 
  3 &  60 & 180 & 100 & 23 &  37 & 110 &  80 \\  
  4 &  38 & 110 &  90 & 24 &  48 & 130 &  90 \\ 
  5 &  19 & 100 &  70 & 25 &  40 & 160 & 112 \\  
  6 &  50 & 170 & 100 & 26 &  36 & 150 &  90 \\   
  7 &  32 & 130 &  84 & 27 &  39 & 140 & 100 \\ 
  8 &  41 & 120 &  80 & 28 &  38 & 110 &  74 \\  
  9 &  36 & 140 &  84 & 29 &  16 & 110 &  70 \\ 
  10 &  57 & 170 & 106 & 30 &  48 & 130 & 100 \\ 
  11 &  52 & 110 &  80 & 31 &  22 & 120 &  80 \\ 
  12 &  19 & 120 &  80 &  32 &  30 & 110 &  70 \\ 
  13 &  17 & 110 &  70 &  33 &  19 & 120 &  80 \\ 
  14 &  16 & 120 &  80 & 34 &  39 & 124 &  84 \\ 
  15 &  67 & 160 &  90 & 35 &  38 & 130 &  94 \\  
  16 &  42 & 130 &  90 & 36 &  45 & 120 &  84 \\  
  17 &  44 & 140 &  90 & 37 &  22 & 130 &  80 \\  
  18 &  56 & 170 & 100 & 38 &  20 & 120 &  86 \\  
  19 &  32 & 150 &  94 &  39 &  18 & 120 &  80 \\ 
  20 &  21 & 140 &  94 & 40 &  31 & 112 &  80 \\  
   \hline
\end{tabular}
\caption{Systolic and Diastolic Blood Pressure of Marwari females in Kolkata}
\label{tab:Table 2}
\end{table}
Our objective is to model the relationship between the systolic and diastolic blood pressures and age for a normal Marwari female living in Kolkata. As demonstrated by \cite{chakraborty1999multivariate} (Also in Figure \ref{fig: fig1}, 
\begin{figure}[ht]
\centering     
\subfigure[Systolic Pressure]{\scalebox{0.3}{\includegraphics{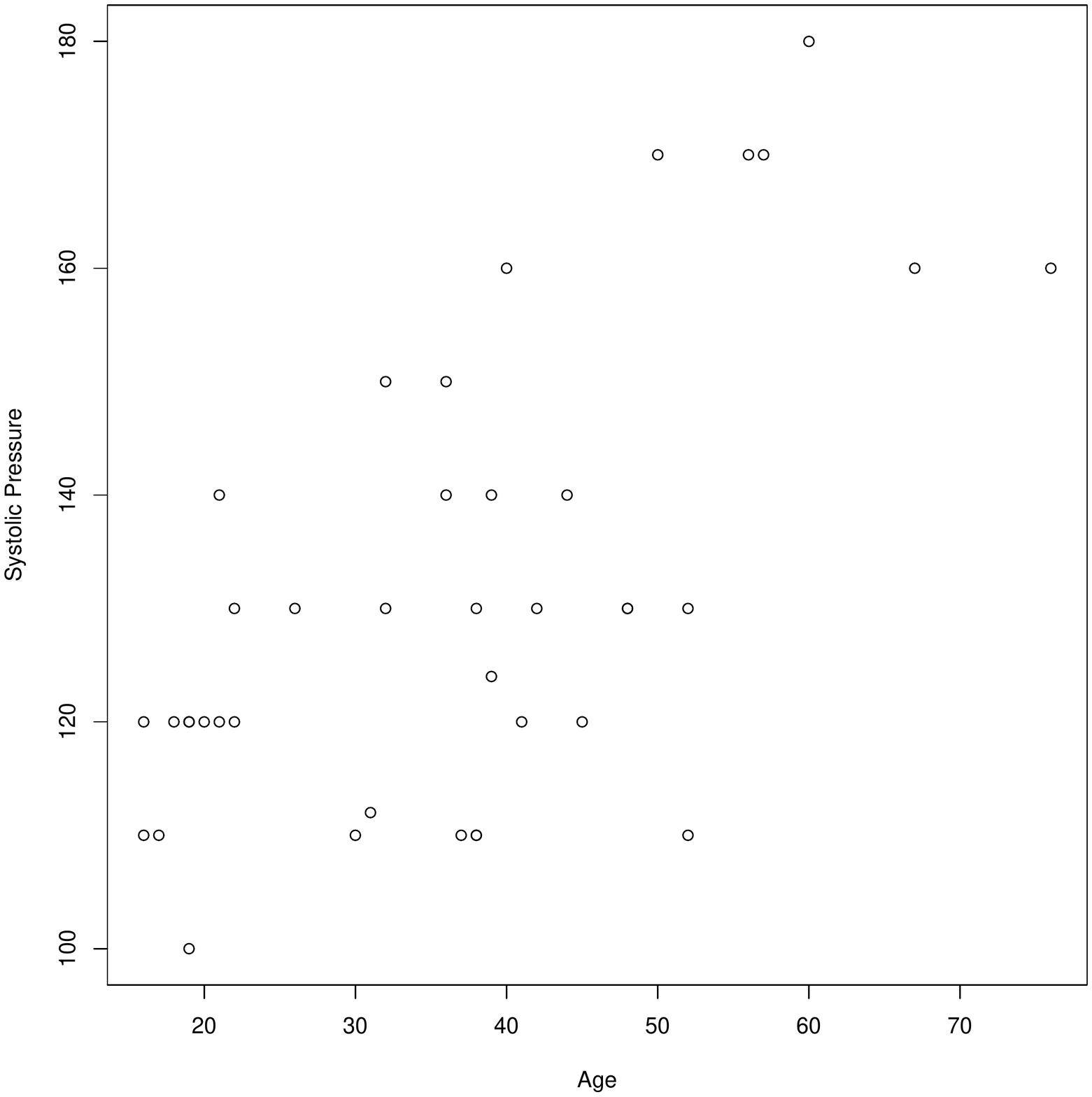}}}
\subfigure[Diastolic Pressure]{\scalebox{0.3}{\includegraphics{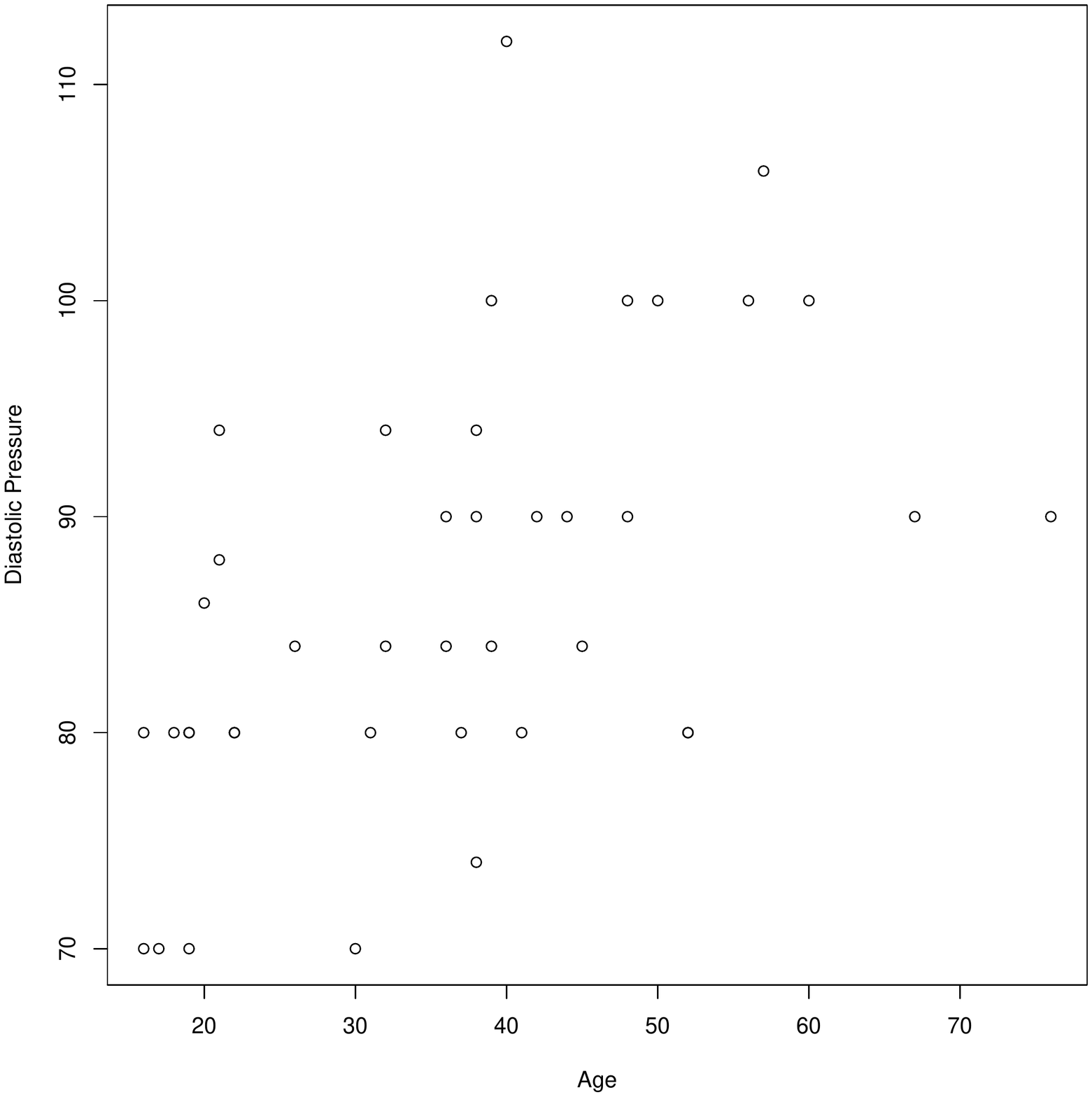}}}
\caption{The Systolic and Diastolic Blood Pressure Against Age of 40 Marwari Females in Kolkata, India.}
\label{fig: fig1}
\end{figure}
which shows the scatterplot of the blood pressure values against age), the data has very high spread and a few outliers. Hence the mean regression is not very efficient, because the mean is sensitive to outliers. Thus a median regression would be appropriate in this situation.

We again model the conditional spatial median of $Y$=(Systolic Pressure ($Y_1$), Diastolic Pressure ($Y_2$)) against $X$=age. For computing the $\delta$-smoothed posterior for the conditional distribution of $Y$ given $X$, we need to sample from the distribution of $X$ truncated in the interval $[x-\delta,x+\delta]$ for every value $x$ of $X$. To choose $c_0$ and $c_1$, we run linear regressions $Y_1$ on $X$ and $Y_2$ on $X$ separately. Of course, the distribution of $X$ is also unknown, so we estimate the density of $X$ using a Gaussian kernel. Based on the regression coefficients, we choose $c_0=(100,73)$ and $c_1=(0.8,0.35)$. We draw 20000 samples from the posterior distribution with a burn-in of 1000, and in Table \ref{tab:tab3}, we show the estimated spatial medians for selected age values along with their coordinate-wise 95\% credible intervals, constructed from the posterior samples. In Figure \ref{fig:2}, we plot the conditional quantiles as a function of $X$, and it shows that both functions increase with $X$, which is expected. In Figure \ref{fig: fig2}, we plot the two components of the spatial median against each other, with $X$=age as labels. This plot also shows an upward trend as well, i.e., diastolic pressure increases as the systolic pressure increases. It also shows that both pressures increase with age, which supports our conclusion from Figure $\ref{fig:2}$.
\begin{table}[ht]
\centering
\begin{tabular}{rrrr}
  \hline
 {Serial No.}& Age & Spatial Median $(Y_1)$ & Spatial Median $(Y_2)$  \\ 
  \hline
  2 & 21 & 120.18 (118.95,\ 121.39) & 90.50 (89.41,\ 91.66) \\ 
  3 & 26 & 130.42 (128.91,\ 131.80) & 100.79 (99.44,\ 102.01)\\ 
  4 & 31 & 116.92 (115.86,\ 118.32)  & 89.48 (86.07, 88.65)\\ 
  5 & 36 & 116.85 (115.91,\ 117.74) & 87.25 (85.84, 88.73)\\ 
  6 & 41 & 115.66 (114.55,\ 116.84) & 85.95 (84.76, 87.04)\\ 
  7 & 46 & 132.10 (131.08,\ 133.39) & 85.24 (84.65,  85.92)\\ 
  8 & 51 & 140.62 (139.60,\ 141.91) & 90.05 (89.46, 90.74)\\ 
  9 & 56 & 117.79 (116.56,\ 118.97) & 87.22 (86.02, 88.31)\\ 
  10 & 61 & 152.30 (151.28,\ 153.59) & 96.19 (95.59, 96.87)\\ 
  11 & 66 & 153.89 (152.87, 155.18) & 93.76 (93.17, 94.44)\\ 
  12 & 71 & 160.91 (159.89, 162.20) & 96.45 (95.86, 97.13)\\
  13 & 76 & 163.19 (162.18, 164.49) & 95.57 (94.97, 96.25)\\
   \hline
\end{tabular}
\caption{Spatial medians for Systolic $(Y_1)$ and Diastolic $(Y_2)$ pressure along with their coordinate-wise 95\% credible intervals (in parenthesis) against selected $x$ age values of Marwari females in Kolkata.}
\label{tab:tab3}
\end{table}
\begin{figure}[ht]
\centering     
\subfigure[Systolic Pressure $(Y_1)$]{\scalebox{0.3}{\includegraphics{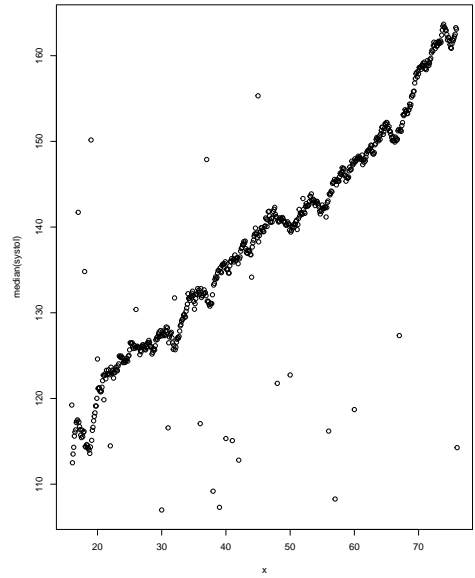}}}
\subfigure[Diastolic Pressure $(Y_2)$]{\scalebox{0.3}{\includegraphics{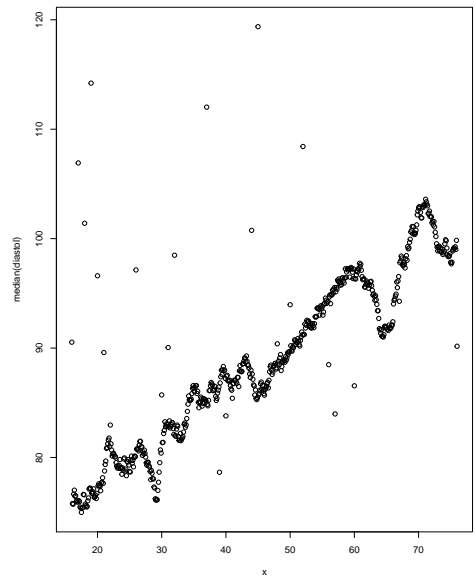}}}
\caption{Conditional spatial median for Systolic pressure $(Y_1)$ and Diastolic pressure $(Y_2)$ as a function of age $(X)$.}
\label{fig:2}
\end{figure}

\begin{figure}[ht]
\includegraphics[scale=0.5,height=80mm,width=80mm]{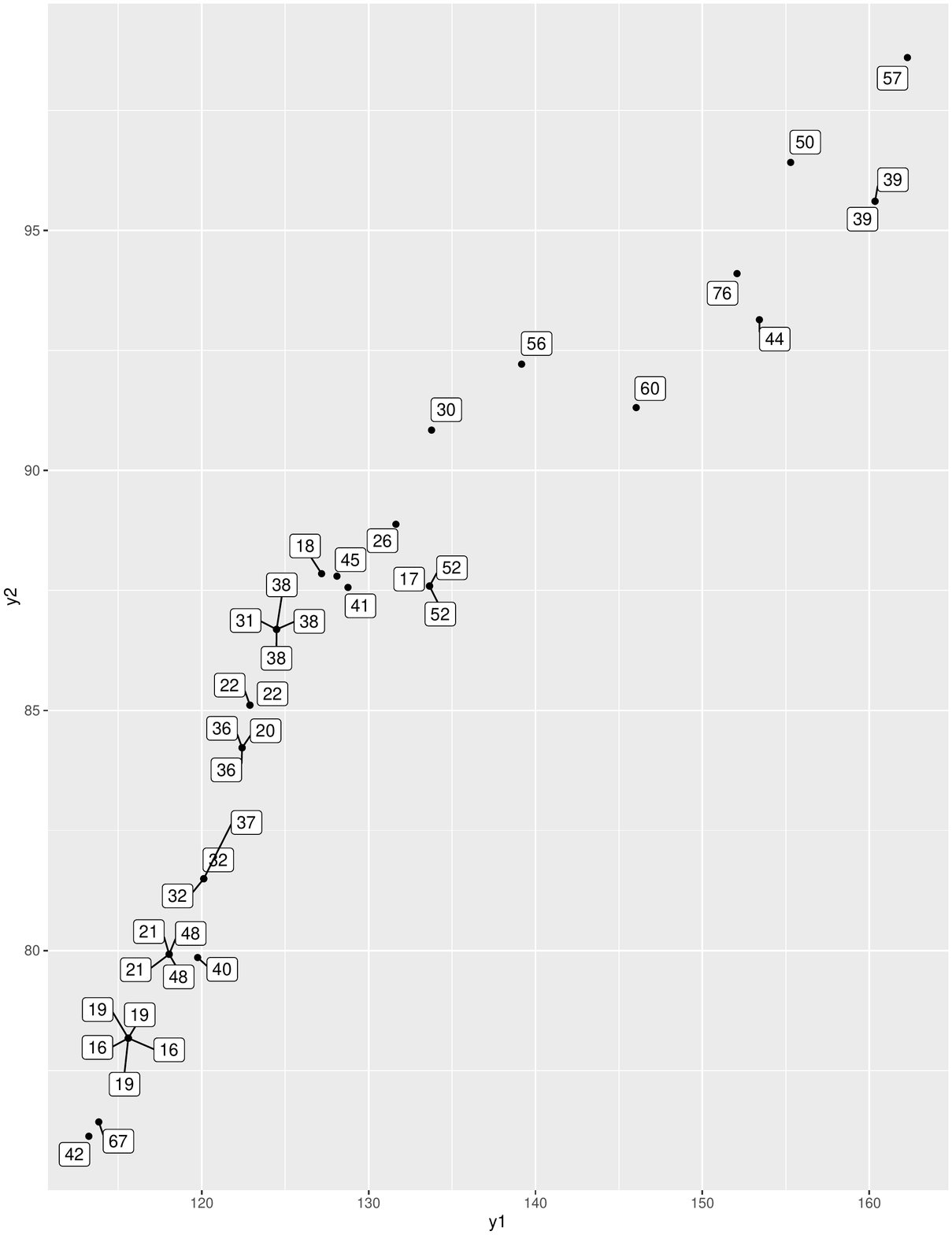}
\caption{Conditional spatial medians $(y_1,y_2)$ for Systolic and Diastolic Blood Pressure of 40 Marwari Females in Kolkata, India, with their age as labels.}
\label{fig: fig2}
\end{figure}
\begin{remark}
Here we have considered Bayesian non-parametric quantile regression of a $k$-dimensional response on a univariate predictor, but the method can be extended to a general $m$-dimensional covariate as well. A DDP prior can be constructed in the same way, and a block Gibbs sampler algorithm can be used, but it would be a lot more computationally extensive.
\end{remark}
\begin{remark}
Here we have considered multivariate quantile regression for geometric quantiles which are obtained by minimizing $P_{Y\vert x}\{\Phi_2(u,Y-\theta)-\Phi_2(u,Y)\}$ with $u \in B_2^{(k)}$. The method can be extended to a more general version of geometric quantiles with general $\ell_p$-norm for $p>1$. 
\end{remark}
\section{Proofs}
\label{2S8}
\begin{proof}[Proof of Lemma 1]
The neighborhood in \eqref{eq455} can be written as
\begin{equation}
    \bigg\{G \times Q :\sup_x \bigg\vert \int \gamma(\xi,\eta,\sigma)(x)\mathrm{d}G(\xi)\mathrm{d}Q(\eta,\sigma)-\int \gamma(\xi,\eta,\sigma)(x)\mathrm{d}G^\star(\xi)\mathrm{d} Q^\star(\eta,\sigma)\bigg\vert <\epsilon\bigg\},
\end{equation}
where $\displaystyle\gamma(\xi,\eta,\sigma)(x)=\int g(y)\sigma^{-k}\phi_k\big(\frac{y-\xi(x)-\eta}{\sigma}\big)\mathrm{d}y$. Note that 
\begin{equation}
\label{eq582}
\begin{split}
    \int \gamma(\xi,\eta,\sigma)(x)\mathrm{d}G(\xi)\mathrm{d}Q(\eta,\sigma)=G(D_1)&Q(D_2)\int \gamma(\xi,\eta,\sigma)(x)\mathrm{d}G_{D_1}(\xi)\mathrm{d}Q_{D_2}(\eta,\sigma)\\
    & +\int_{(D_1\times D_2)^c}\gamma(\xi,\eta,\sigma)(x)\mathrm{d}G(\xi)\mathrm{d}Q(\eta,\sigma),
    \end{split}
\end{equation}
where $G_{D_1}$ denotes the measure $G$ restricted and normalized to a compact set $D_1 \subset \mathcal{C}(\mathfrak{X})$, and $Q_{D_2}$ denotes the measure $Q$ restricted and normalized to the compact set $D_2 \subset \mathbb{R}^k\times \mathbb{R}^+$. For every $x \in \mathfrak{X}$, the second term on the right hand side of \eqref{eq582} can be bounded above by $(G\times Q)(D_1\times D_2)^c$. Then
\begin{equation}
\label{eq583}
\begin{split}
    \sup_x &\bigg\vert \int \gamma(\xi,\eta,\sigma)(x)\mathrm{d}G(\xi)\mathrm{d}Q(\eta,\sigma)-\int \gamma(\xi,\eta,\sigma)\mathrm{d}G^\star(\xi)\mathrm{d}Q^\star(\eta,\sigma)\bigg\vert \\
    &\leq \sup_x\bigg\vert \int \gamma(\xi,\eta,\sigma)(x)\mathrm{d}G_{D_1} (\xi)\mathrm{d}Q_{D_2}(\eta,\sigma)-\int \gamma(\xi,\eta,\sigma)(x)\mathrm{d}G^\star(\xi)\mathrm{d}Q^\star(\eta,\sigma)\bigg\vert \\
    &+\bigg\vert \frac{1}{G(D_1)Q(D_2)}-1\bigg\vert \sup_x \bigg\vert \int \gamma(\xi,\eta,\sigma)(x)\mathrm{d}G(\xi)\mathrm{d}Q(\eta,\sigma)\bigg\vert +(G\times Q)(D_1\times D_2)^c\\
    & \leq \sup_x\bigg\vert \int \gamma(\xi,\eta,\sigma)(x)\mathrm{d}G_{D_1} (\xi)\mathrm{d}Q_{D_2}(\eta,\sigma)-\int \gamma(\xi,\eta,\sigma)(x)\mathrm{d}G^\star(\xi)\mathrm{d}Q^\star(\eta,\sigma)\bigg\vert \\
    &+2 \frac{(G\times Q)(D_1\times D_2)^c}{(G\times Q)(D_1\times D_2)}.
\end{split}
\end{equation}
The following lemma (Lemma \ref{l581}) says that the family of functions $(\xi,\eta,\sigma)\mapsto \\\{\gamma(\xi,\eta, \sigma)(x):x \in \mathfrak{X}\}$ is uniformly bounded and equicontinuous. Hence by the Arzela-Ascoli theorem, the family is pre-compact, and hence totally bounded. Hence for any $\epsilon>0$, there exist $x_1,\dots,x_s \in \mathfrak{X}$, such that for every $x \in \mathfrak{X}$, there exist $i=1,\dots,s$ such that 
\begin{equation}\label{eq584}
    \vert \gamma(\xi,\eta,\sigma)(x)-\gamma(\xi,\eta,\sigma)(x_i)\vert < \epsilon,
\end{equation}
for every $(\xi,\eta,\sigma)\in D$. Hence for every $x \in \mathfrak{X}$
\begin{equation}
\label{eq585}
\begin{split}
    \bigg\vert & \int \gamma(\xi,\eta,\sigma)(x)\mathrm{d}G_{D_1}(\xi)\mathrm{d}Q_{D_2}(\eta,\sigma)-\int \gamma(\xi,\eta,\sigma)(x)\mathrm{d}G^\star(\xi)\mathrm{d}Q^\star(\eta,\sigma)\bigg\vert \\
    &\leq \bigg\vert \int \gamma(\xi,\eta,\sigma)(x)\mathrm{d}G_{D_1}(\xi)\mathrm{d}Q_{D_2}(\eta,\sigma)-\int \gamma(\xi,\eta,\sigma)(x_i)\mathrm{d}G_{D_1}(\xi)\mathrm{d}Q_{D_2}(\eta,\sigma)\bigg\vert\\
    & + \bigg\vert \int \gamma(\xi,\eta,\sigma)(x_i)\mathrm{d}G_{D_1}(\xi)\mathrm{d}Q_{D_2}(\eta,\sigma)-\int \gamma(\xi,\eta,\sigma)(x_i)\mathrm{d}G^\star(\xi)Q^\star(\eta,\sigma)\bigg\vert\\
    &+\bigg\vert\int \gamma(\xi,\eta,\sigma)(x_i)\mathrm{d}G^\star(\xi)\mathrm{d}Q^\star(\eta,\sigma)-\int \gamma(\xi,\eta,\sigma)(x)\mathrm{d}G^\star(\xi)Q^\star(\eta,\sigma)\bigg\vert.
    \end{split}
\end{equation}
The first and third terms in the right hand side of \eqref{eq584} are bounded above by $\epsilon$, using \eqref{eq583}. For the second term, note that, since $G^\star(D_1)=Q^\star(D_2)=1$, for every $\epsilon >0$, there exists a weak neighborhood $\mathcal{W}_1^\star$ of $G^\star$ and $\mathcal{W}_2^\star$ of $Q^\star$ respectively in $\mathfrak{M}(\mathcal{C}(\mathfrak{X}))$ and $\mathfrak{M}(\mathbb{R}^k\times \mathbb{R}^+)$ such that for every $G \in \mathcal{W}_1^\star$ and $Q \in \mathcal{W}_2^\star$, $G(D_1)>1-\epsilon$, and $Q(D_2)>1-\epsilon$, and for all $i=1,\dots,s$,
\begin{equation}
    \bigg\vert \int \gamma(\xi,\eta,\sigma)(x_i)\mathrm{d}G(\xi)\mathrm{d}Q(\eta,\sigma)-\int \gamma(\xi,\eta,\sigma)(x_i)\mathrm{d}G^\star(\xi)\mathrm{d}Q^\star(\eta,\sigma)\bigg\vert<\epsilon.
\end{equation}
Then for $G\in \mathcal{W}_1^\star$, and $Q \in \mathcal{W}_2^\star$, 
\begin{align*}
    \bigg\vert \int &\gamma(\xi,\eta,\sigma)(x_i)\mathrm{d}G_{D_1}(\xi)\mathrm{d}Q_{D_2}(\eta,\sigma)-\int \gamma(\xi,\eta,\sigma)(x_i)\mathrm{d}G^\star(\xi)\mathrm{d}Q^\star(\eta,\sigma)\bigg\vert\\
    &\leq \bigg\vert \frac{1}{G(D_1)Q(D_2)}-1\bigg\vert + \epsilon\leq \epsilon + \frac{1-(1-\epsilon)^2}{(1-\epsilon)^2} \leq 4\epsilon,
\end{align*}
if $\epsilon < 1-\sqrt{3}/2$. Therefore, the right hand side in \eqref{eq583} is less than $6\epsilon$. Plugging everything in \eqref{eq583}, the right hand side of it can be bounded above by $10\epsilon$. Thus, to show that the left hand side of \eqref{eq583} has positive prior probability, all we need to show is any weak neighborhood $\mathcal{W}_1^\star$ of $G^\star$ and $\mathcal{W}_2^\star$ of $Q^\star$ have positive prior probability. The measure $G$ has a $\mathrm{DP}(M_1G_0)$ prior with $G_0$ being a Gaussian process, having full support on $\mathcal{C}(\mathfrak{X})$. Similarly, $Q$ has a $\mathrm{DP}(M_2Q_0)$ prior, where $Q_0$ is the product measure of a $k$-dimensional Gaussian and an inverse gamma distribution, which also has a full support on $\mathbb{R}^k\times \mathbb{R}^+$. Thus, by Lemma 3.6 in \cite{ghosal2017fundamentals}, the weak neighborhoods have positive prior probability.
\end{proof}
\begin{lemma}{\label{l581}}
Define $\displaystyle\gamma(\xi,\eta,\sigma)(x)=\int g(y)\sigma^{-k}\phi_k\big(\frac{y-\xi(x)-\eta}{\sigma}\big)\mathrm{d}y$, where $g: \mathbb{R}^k \to [0,1]$ is bounded and continuous. Then the family of maps $(\xi,\eta,\sigma) \mapsto \{\gamma(\xi,\eta,\sigma)(x): x \in \mathfrak{X}\}$ is uniformly equicontinuous as a family of functions of $(\xi,\eta,\sigma)$ on the compact metric space $D=D_1\times D_2$, i.e., for all $x \in \mathfrak{X}$, and all $\Vert (\xi,\eta,\sigma)-(\xi^{\prime},\eta^{\prime}, {\sigma^{\prime}})\Vert < \delta$, we have 
\begin{equation}
\label{eq587}
\vert \gamma(\xi,\eta,\sigma)(x)-\gamma(\xi^\prime,\eta^\prime,\sigma^\prime)(x) \vert < \epsilon.
\end{equation} 
\end{lemma}
\begin{proof}
For this proof, we borrow some ideas from the proof of Theorem 3 in \cite{ghosal1999posterior}. Using the fact that $0 \leq g(\cdot) \leq 1$, for each $x \in \mathfrak{X}$, the left hand side of \eqref{eq587} can be bounded as
\begin{equation}\label{eq588}
\begin{split}
    \bigg\vert \int g(y)\sigma^{-k}&\phi_k\bigg(\frac{y-\xi(x)-\eta}{\sigma}\bigg)\mathrm{d}y-\int g(y){\sigma^\prime}^{-k}\phi_k\bigg(\frac{y-\xi^\prime(x)-\eta^\prime}{\sigma^\prime}\bigg)\mathrm{d}y\bigg\vert\\
    &\leq \Vert \xi(x)+\eta-\xi^\prime(x)-\eta^\prime \Vert +\vert \sigma-\sigma^\prime \vert\\
    &\leq \Vert \xi -\xi^\prime \Vert_\infty + \Vert \eta-\eta^\prime \Vert_2+\vert \sigma-\sigma^\prime \vert.
    \end{split}
\end{equation}
The last inequality follows from the Lipschitz continuity of $\phi_k(\cdot)$ as a function of $(\xi,\eta,\sigma)$, which gives us the conclusion.
\end{proof}
\begin{proof}[Proof of Lemma \ref{l1}]
Note that $\mathrm{KL}(h^\star,h)$ can be decomposed as
\begin{equation}{\label{eq3}}
\begin{aligned}
    \mathrm{KL}(h^\star,h)=
    \int_{\mathfrak{X}}\int_{\mathcal{K}}f^\star(y\vert x)\log\frac{f^\star(y\vert x)}{f(y \vert x)}\mathrm{d}yq(x)\mathrm{d}x+ \int_{\mathfrak{X}}\int_{\mathcal{K}^c}f^\star(y\vert x)\log\frac{f^\star(y\vert x)}{f(y \vert x)}\mathrm{d}yq(x)\mathrm{d}x,
    \end{aligned}
\end{equation}
where $\mathcal{K}=\{y:\Vert y \Vert_2 \leq K\}$, for some $K>0$. First, we show that the second term in the RHS of \eqref{eq3} is sufficiently small. Note that
\begin{align*}
  \int_{\mathfrak{X}}\int_{ \mathcal{K}^c}f^\star(y\vert x)\log\frac{f^\star(y\vert x)}{f(y \vert x)}\mathrm{d}y & q(x)\mathrm{d}x\\
  \leq   \int_{\mathfrak{X}}\int_{\mathcal{K}^c}f^\star(y\vert x)&\log\frac{{\underset{(\xi,\eta,\sigma) \in D}{\sup}\frac{1}{\sigma^k}\phi_k\big(\frac{y-\xi(x)-\eta}{\sigma}\big)}}{\underset{(\xi,\eta,\sigma) \in D}{\inf}\frac{1}{\sigma^k}\phi_k(\frac{y-\xi(x)-\eta}{\sigma})G^\star(D_1)Q^\star(D_2)}\mathrm{d}yq(x)\mathrm{d}x,
\end{align*}
where $D_1$ and $D_2$ are compact metric spaces. For $\xi \in D_1$, $\displaystyle\sup_x \Vert \xi(x) \Vert < b^\star$ for some $b^\star >0$. Also, for $(\eta,\sigma) \in D_2$, $\Vert \eta \Vert_2 < a^\star$, and $\barbelow{\sigma} <\sigma< \bar{\sigma}$, for $a^\star,\barbelow{\sigma},\bar{\sigma}>0$. For $\Vert y \Vert_2 > K > a^\star+b^\star$,
$$
\log\inf_{(\xi,\eta,\sigma) \in D}\frac{1}{\sigma^k}\phi_k\bigg(\frac{y-\xi(x)-\eta}{\sigma}\bigg)=\log\bigg\{\frac{1}{\barbelow{\sigma}^k}\phi_k\bigg(\frac{y+(a^{\star}+b^{\star})\frac{y}{\Vert y \Vert}}{\barbelow{\sigma}^k}\bigg)\bigg\}.
$$
Let $\mathscr{V}=\{G \times Q: (G\times Q) (D)>\barbelow{\sigma}^k/\bar{\sigma}^k\}$. Since $(G^\star \times Q^\star)(D)=1$, and $D$ is an open set, $\mathscr{V}$ contains a neighborhood of $G^\star \times Q^\star$ of the form \eqref{eq455}. Thus for every $G \times Q \in \mathscr{V}$,

\begin{align*}
    \int_{\mathfrak{X}}&\int_{\mathcal{K}^C}  f^\star(y\vert x)\log \frac{f^\star(y\vert x)}{f(y\vert x)}\mathrm{d}yq(x)\mathrm{d}x \\
    &\leq \int_{\mathfrak{X}}\int_{ \mathcal{K}^c}  f^\star(y \vert x)\log \frac{{\underset{(\xi,\eta,\sigma) \in D}{\sup}\phi_k\big(\frac{y-\xi(x)-\eta}{\sigma}\big)}}{\underset{(\xi,\eta,\sigma )\in D}{\inf}\phi_k\big(\frac{y-\xi(x)-\eta}{\sigma}\big)}\mathrm{d}yq(x)\mathrm{d}x\\
    &=\int_{\mathfrak{X}}\int_{\mathcal{K}^c}f^\star(y\vert x)\log \frac{\phi_k\big(\frac{y+(a^\star+b^\star)\frac{y}{\Vert y \Vert_2}}{\bar{\sigma}}\big)}{\phi_k\big(\frac{y-(a^\star+b^\star)\frac{y}{\Vert y \Vert_2}}{\barbelow{\sigma}}\big)}\mathrm{d}yq(x)\mathrm{d}x\\
    =&\int_{\mathfrak{X}}\int_{\mathcal{K}^c}\Big\{-\frac{1}{2\bar{\sigma}^k}\Big\Vert y+(a^\star+b^\star)\frac{y}{\Vert y \Vert_2} \Big\Vert_2^2+\frac{1}{2\bar{\sigma}^k}\Big\Vert y-(a^\star+b^\star )\frac{y}{\Vert y \Vert_2}\Big\Vert_2^2\}f^\star(y\vert x)\mathrm{d}yq(x)\mathrm{d}x\\
    =& \int_{\mathfrak{X}}\int_{\mathcal{K}^c} \{-\frac{1}{2\bar{\sigma}^k}\Big\Vert \Vert y \Vert_2+(a^\star+b^\star) \big\Vert_2^2+\frac{1}{2\bar{\sigma}^k}\big\Vert \Vert y\Vert_2-(a^\star+b^\star )\Big\Vert_2^2\}f^\star(y\vert x)\mathrm{d}yq(x)\mathrm{d}x.
\end{align*}
Since $f^\star$ is of the form \eqref{eq72}, with $\xi$ being uniformly bounded on $D_1$, and $(\eta,\sigma)$ being bounded on $D_2$, for every $\epsilon > 0$, we can find a compact set $\mathcal{K}$ such that 
\begin{equation}
 \int_{\mathfrak{X}}\int_{ \mathcal{K}^C}  f^\star(y\vert x)\log \frac{f^\star(y\vert x)}{f(y\vert x)}\mathrm{d}yq(x)\mathrm{d}x < \frac{\epsilon}{2}.   
\end{equation}
Next, we show that,
   \begin{equation}
 \int_{\mathfrak{X}}\int_{\mathcal{K}}  f^\star(y\vert x)\log \frac{f^\star(y\vert x)}{f(y\vert x)}\mathrm{d}yq(x)\mathrm{d}x < \frac{\epsilon}{2}.    
\end{equation}
Following the arguments in Lemma \ref{l581}, it can be shown that the family of maps $(\xi,\eta,\sigma) \mapsto \big\{\frac{1}{\sigma^k}\phi_k\big(\frac{y-\xi(x)-\eta}{\sigma}\big):y \in \mathcal{K},x\in \mathfrak{X})\big\}$ is uniformly equicontinuous on $D$. Thus, the family is uniformly bounded on $D$, and pre-compact by Arzela-Ascoli theorem. Hence there exist $x_i,y_i,\ i=1,\dots,m$ such that, for any $y \in \mathcal{K}$, $x \in \mathfrak{X}$,
\begin{equation}\label{eq78}
\underset{(\xi,\eta,\sigma) \in D}{\sup} \big\vert \sigma^{-k}\phi_k(\frac{y-\xi(x)-\eta}{\sigma})-\sigma^{-k}\phi_k(\frac{y_i-\xi(x_i)-\eta}{\sigma})\big\vert < c^\star\delta,
\end{equation}
where $\displaystyle c^\star=\sup_{x \in \mathfrak{X},y \in \mathcal{K}}\sup_{(\xi,\eta,\sigma)\in D}\big\vert \sigma^{-k}\phi_k\big(\frac{y_i-\xi(x_i)-\eta}{\sigma}\big)\big\vert$. Define 
\begin{equation}
    \begin{split}
\mathscr{U}=\{G \times Q:\vert & \int_D \sigma^{-k}\phi_k\big(\frac{y-\xi(x)-\eta}{\sigma})\mathrm{d}G^{\star}(\xi\big)\mathrm{d}Q^\star(\eta,\sigma)-\\
&\int_D\sigma^{-k}\phi_k\big(\frac{y_i-\xi(x_i)-\eta}{\sigma}\big)\mathrm{d}G(\xi)\mathrm{d}Q(\eta,\sigma)\vert <c^\star\delta,\ i=1,\dots,m\}.
\end{split}
\end{equation}
Then $\mathscr{U}$ is a finite intersection of neighborhoods of $G^\star \times Q^\star$ of the form \eqref{eq455}. Since $\mathrm{supp}(G \times Q) \subset D$, 
\begin{align*}
  \int_{\mathfrak{X}}\int_{ \mathcal{K}}  f^\star(y\vert x)&\log \frac{f^\star(y\vert x)}{f(y\vert x)}\mathrm{d}yq(x)\mathrm{d}x <\\
  &\int_{\mathfrak{X}}\int_{ \mathcal{K}}  f^\star(y\vert x)\log \frac{\int_D \phi_k (\frac{y-\xi(x)-\eta}{\sigma})\mathrm{d}G^\star(\xi)\mathrm{d}Q^\star(\mu,\sigma)}{\int_D \phi_k (\frac{y-\xi(x)-\eta}{\sigma})\mathrm{d}G(\xi)\mathrm{d}Q(\mu,\sigma)}\mathrm{d}yq(x)\mathrm{d}x .
\end{align*}
Without loss of generality, we assume $(G^\star \times Q^\star)(\partial D)=0$, where $\partial X$ denotes the boundary of the set $X$. For any $(G\times Q) \in \mathscr{U}$, $y \in \mathcal{K}$ and $x \in \mathfrak{X}$, denoting $\displaystyle g_1(y,x,\xi,\eta,\sigma)=\sigma^{-k}\phi_k \bigg(\frac{y-\xi(x)-\eta}{\sigma}\bigg)$,
\begin{equation}\label{eq77}
\begin{aligned}
    &\bigg\vert \int_D g_1(y,x,\xi,\eta,\sigma)\mathrm{d}G^{\star}(\xi)\mathrm{d}Q^\star(\eta,\sigma)-
   \int_D g_1(y,x,\xi,\eta,\sigma)\mathrm{d}G(\xi)\mathrm{d}Q(\eta,\sigma)\bigg\vert\\
   &\leq \bigg\vert \int_D g_1(y,x,\xi,\eta,\sigma)\mathrm{d}Q^\star(\eta,\sigma)-
   \int_{D} g_1(y_i,x_i,\xi,\eta,\sigma)\mathrm{d}G^{\star}(\xi)\mathrm{d}Q^\star(\eta,\sigma)\bigg\vert\\
  & + \bigg\vert \int_D g_1(y_i,x_i,\xi,\eta,\sigma)\mathrm{d}G^{\star}(\xi)\mathrm{d}Q^\star(\eta,\sigma)-
   \int_D g_1(y_i,x_i,\xi,\eta,\sigma)\mathrm{d}G(\xi)\mathrm{d}Q(\eta,\sigma)\bigg\vert\\
   &+ \bigg\vert \int_D g_1(y_i,x_i,\xi,\eta,\sigma)\mathrm{d}G(\xi)\mathrm{d}Q(\eta,\sigma)-
   \int_D g_1(y,x,\xi,\eta,\sigma)\mathrm{d}G(\xi)\mathrm{d}Q(\eta,\sigma)\bigg\vert
\end{aligned}
\end{equation}
The first and third terms on the right hand side of \eqref{eq77} are each less than $c^\star\delta$ by \eqref{eq78}. The second term is also less than $c^\star\delta$, since $(G\times Q) \in \mathscr{U}$. Thus,
\begin{equation}
    \begin{aligned}
    \bigg\vert \int_D \frac{1}{\sigma^k} \phi_k\bigg(\frac{y-\xi(x)-\eta}{\sigma}\bigg)\mathrm{d}G^{\star}(\xi)\mathrm{d}Q^\star(\eta,\sigma)-
   \int_D \frac{1}{\sigma^k} \phi_k\bigg(\frac{y-\xi(x)-\eta}{\sigma}\bigg)\mathrm{d}G(\xi)\mathrm{d}Q(\eta,\sigma)\bigg\vert\\
   < 3c^\star\delta.
    \end{aligned}
\end{equation}
Therefore, for $(G \times Q) \in \mathscr{U}$,
\begin{equation}
    \bigg\vert\frac{\int_D \frac{1}{\sigma^k}\phi_k\big(\frac{y-\xi(x)-\eta}{\sigma}\big)\mathrm{d}G^\star(\xi)\mathrm{d}Q^\star(\eta,\sigma)}{\int_D \frac{1}{\sigma^k}\phi_k\big(\frac{y-\xi(x)-\eta}{\sigma}\big)\mathrm{d}G(\xi)\mathrm{d}Q(\eta,\sigma)}-1\bigg\vert
< \frac{3\delta}{1-3\delta},
\end{equation}
for $\delta < \frac{1}{3}$. Thus, by choosing $\delta$ small enough
\begin{align*}
    \int_{\mathfrak{X}}\int_{ \mathcal{K}}  f^\star(y\vert x)&\log \frac{f^\star(y\vert x)}{f(y\vert x)}\mathrm{d}yq(x)\mathrm{d}x \leq \\ &\sup_{x\in \mathfrak{X},y \in \mathscr{K}}\bigg\vert\frac{\int_D \frac{1}{\sigma^k}\phi_k\big(\frac{y-\xi(x)-\eta}{\sigma}\big)\mathrm{d}G^\star(\xi)\mathrm{d}Q^\star(\eta,\sigma)}{\int_D \frac{1}{\sigma^k}\phi_k\big(\frac{y-\xi(x)-\eta}{\sigma}\big)\mathrm{d}G(\xi)\mathrm{d}Q(\eta,\sigma)}-1\bigg\vert
    < \frac{\epsilon}{2},
    \end{align*}
for $G\times Q \in \mathscr{U}$. Thus for any $\epsilon>0$ and $G \times Q \in \mathscr{V} \cap \mathscr{U}$,
\begin{equation}
     \int_{\mathfrak{X}}\int_{\mathbb{R}^k}  f^\star(y\vert x)\log \frac{f^\star(y\vert x)}{f(y\vert x)}\mathrm{d}yq(x)\mathrm{d}x < \epsilon.
\end{equation}
Thus Lemma \ref{l1} is proved.
\end{proof}
\begin{proof}[Proof of Theorem \ref{th1}]
Using Example 6.20 in \cite{ghosal2017fundamentals}, Theorem \ref{th1} implies that the posterior is weakly consistent at $f^\star$, i.e., for any $W_{\epsilon,g}(f^\star)$
\begin{equation}
    \Pi\{W_{\epsilon,g}(f^\star)\vert (Y^n,X^n)\} \rightarrow 1.
\end{equation}
The above fact further implies that, for any $\epsilon >0$
\begin{equation}\label{eq4812}
   \Pi\{ \sup_{x,y}\vert F_{X,Y}(x,y)-F_{X,Y}^\star(x,y)\vert<\epsilon \vert (Y^n,X^n)\}\rightarrow 1,
\end{equation}
Thus there exists $\epsilon_n \downarrow 0$ such that \eqref{eq4812} holds with $\epsilon_n$ replacing $\epsilon$. Note that for any $\epsilon>0$, $\delta>0$ and $F$ such that $\sup_{x,y} |F_{X,Y}(x,y)-F_{X,Y}^\star(x,y)|<\epsilon$, we have  
$$
\sup_{x,y}\bigg\vert \frac{P_{X,Y}(\vert X-x \vert \leq \delta, Y \leq y) }{P_X(\vert X-x \vert \leq \delta)}- \frac{P^\star_{X,Y}(\vert X-x \vert \leq \delta, Y \leq y) }{P_X(\vert X-x \vert \leq \delta)}\bigg\vert < \frac{\epsilon}{P_X(\vert X-x \vert \leq \delta)},
$$
Note that $P_X(\vert X-x\vert \leq \delta)=\int_{x-\delta}^{x+\delta}q(x)\mathrm{d}x \geq 2\delta a$, with $a=\underset{x}{\min}\ q(x)$. Choosing a fixed sequence $\delta_n \downarrow 0$ at a rate slower than $\epsilon_n$, 
$$
\sup_{x,y}\bigg\vert \frac{P_{X,Y}(\vert X-x \vert \leq \delta_n, Y \leq y) }{P_X(\vert X-x \vert \leq \delta_n)}- \frac{P_{X,Y}^\star(\vert X-x \vert \leq \delta_n, Y \leq y) }{P_X(\vert X-x \vert \leq \delta_n)}\bigg\vert < \epsilon_n,
$$
for every $n$. Notice that $\displaystyle \lim_{\delta_n \rightarrow 0} \frac{P_{X,Y}^\star(\vert X-x \vert \leq \delta_n, Y \leq y) }{P_X(\vert X-x \vert \leq \delta_n)}\rightarrow F^\star_{Y\vert x}(y)$. For a fixed $u \in B_2^{(k)}$, note that $Q_{\delta;Y\vert x}(u)$ can be written as $\displaystyle Q_{\delta;Y\vert x}(u)= \argmax_\theta \int g(u;y,\theta)\mathrm{d}F_{\delta; Y\vert x}(y)$, where $g(u;y,\theta)$ is defined as
\begin{equation}
    g(u;y,\theta)=-\{\Vert y-\theta \Vert_2+\langle u,y-\theta \rangle -\Vert y \Vert_2 + \langle u,y \rangle\}.
\end{equation}
Since $g(u;y,\theta)$ is a bounded and continuous function in $y$ for every fixed $\theta \in \mathbb{R}^k$, and $x \in \mathfrak{X}$,
\begin{equation}
\begin{split}
    \Pi\bigg\{\vert\int g(u;y,\theta)\mathrm{d}F_{\delta_n; Y\vert x}(y)-\int g(u;y,\theta)\mathrm{d}F^\star_{Y \vert x}(y)\vert < \epsilon\vert (Y^n,X^n)\bigg\}\rightarrow 1\ a.s.,
    \end{split}
\end{equation}
for every $\epsilon>0$. We use the argmax theorem (Theorem 5.7 in \cite{van2000asymptotic}) to achieve the assertion in Theorem \ref{th1}. We need the following two conditions:
\begin{enumerate}
    \item For every $\epsilon>0$ and fixed $u \in B_2^{(k)}$, and for all $x \in \mathfrak{X}$,
    \begin{equation}
    \begin{aligned}
    \Pi\bigg\{\sup_{\theta}\vert \int g(u;y,\theta)\mathrm{d}F_{\delta_n; Y\vert x}(y)-\int g(u;y,\theta)\mathrm{d}F^\star_{Y\vert x}(y)\vert <\epsilon \vert (Y^n,X^n)\bigg\} \rightarrow 1.
    \end{aligned}
    \end{equation}
    \item $\displaystyle \sup_{\theta: \Vert \theta-\theta^\star\Vert_2 \geq \epsilon} \int g(u;y,\theta)\mathrm{d}F^\star_{Y\vert x}(y) < \int g(u;y,Q^\star_{Y\vert X}(u \vert x))\mathrm{d}F^\star_{Y\vert x}(y)$, which is also known as the \enquote{well-separatedness} condition.  
\end{enumerate}
To prove the above conditions, we need to restrict the parameter space to a compact subset of $\mathbb{R}^k$, which leads us to the following lemma, which says that the parameter space can be taken to be a compact set with high probability.
\begin{lemma}{\label{l5}}
For every $x \in \mathfrak{X}$, and for every fixed $u \in B_2^{(k)}$, for every $0<\epsilon < c^{-1}/(c^{-1}+\Vert u \Vert_2+1)$ and $K_x>0$ such that $P_{Y\vert x}^\star(\Vert Y \Vert_2\leq K_x)>1-\epsilon$, the posterior probability of $Q_{Y \vert x}(u) \leq cK_x$ given $(Y^n,X^n)$ tends to 1, a.s. $n \rightarrow \infty$, where $c=3/(1-\Vert u \Vert_2)$.
\end{lemma}
Proof of Lemma \ref{l5} is given at the end of this proof. Using Lemma \ref{l5}, the parameter space can be taken to be $\Theta$, which is a compact subset of $\mathbb{R}^k$. Condition 1 is proved using Example (A.2) in \cite{bickel1992uniform}. We have to show that, for every $\theta \in \Theta$ with $\Theta$ compact, and $u \in B_2^{(k)}$,
\begin{itemize}
    \item $\displaystyle\sup_y \vert g(u;y,\theta) \vert \leq k_0$
    \item $\displaystyle\sup_y \{\vert g(u;y,\theta)-g(u;y^\prime,\theta)\vert /\Vert y-y^\prime \Vert_2  \} \leq k_0$.
\end{itemize}
The first condition follows from
\begin{align*}
    \vert g(u;y,\theta) \vert \leq \Vert \theta \Vert_2+ \langle u,\theta \rangle \leq 2\Vert \theta \Vert_2 \leq 2cK_x,
\end{align*}
The second condition follows from the Lipschitz continuity of the functions $g(u;y,\theta)$,
\begin{align*}
    \vert g(u;y,\theta)-g(u;y^\prime,\theta) \vert =& \vert \Vert y-\theta \Vert_2-\Vert y^\prime-\theta \Vert_2-\Vert y \Vert_2+\Vert y^\prime \Vert_2 \vert 
    \leq  2\Vert y-y^\prime \Vert_2. 
\end{align*}
Then
$
\sup_y \{\vert g(u;y,\theta)-g(u;y^\prime,\theta)\vert /\Vert y-y^\prime \Vert_2  \} \leq 2.
$
Condition 2 follows from our assumption, which proves Lemma \ref{l1}.
\end{proof}
\begin{proof}[Proof of Lemma \ref{l5}]
Define $M(F^\star_{Y\vert x},\theta)=F^\star_{Y\vert x}\{ \Phi_2(u,Y-\theta)-\Phi_2(u,Y)\}=F^\star_{Y\vert x}(\Vert Y-\theta \Vert_2-\Vert Y \Vert_2-\langle u,\theta \rangle)$. We show that for $0<\epsilon<c^{-1}/(c^{-1}+\Vert u \Vert_2+1)$, there exists $K_x>0$ such that $\Vert \theta \Vert_2 \geq cK_x$ implies $M(F^\star_{Y\vert x},\theta)>0$. If $\Vert Y \Vert_2 \leq K_x$ and $\Vert \theta \Vert_2 \geq cK_x$, then
\begin{equation*}
\Vert Y-\theta \Vert_2 \geq \Vert \theta \Vert_2 - \Vert Y \Vert_2 \geq\frac{(c-1)\Vert \theta \Vert_2}{c}+K-\Vert Y \Vert_2 \geq \frac{\Vert \theta \Vert_2}{c},
\end{equation*}
Hence as $\Vert Y \Vert_2 \leq K_x \leq \Vert \theta \Vert_2/c$,
\begin{equation*}
\Vert Y-\theta \Vert_2-\Vert Y \Vert_2-\langle u,\theta \rangle \geq \frac{(c-1)\Vert \theta \Vert_2}{c}-\frac{\Vert \theta \Vert_2}{c}- \Vert u \Vert_2 \Vert \theta \Vert_2.
\end{equation*}
Using the relation $c=3/(1-\Vert u \Vert_2)$
$$
\Vert Y-\theta \Vert_2-\Vert Y \Vert_2-\langle u,\theta \rangle \geq \frac{\Vert \theta \Vert_2}{c}.
$$
Now since always $\big\vert \Vert Y -\theta \Vert_2 - \Vert Y \Vert_2 -\langle u,\theta \rangle \big\vert \leq (1+\Vert u \Vert_2)\Vert \theta \Vert_2$, we can write
\begin{align*}
M(F_{Y\vert x}^\star,\theta) &= \int_{\Vert Y \Vert_2 \leq K_x}(\Vert Y-\theta \Vert_2-\Vert Y \Vert_2-\langle u,\theta \rangle)\mathrm{d}F_{Y\vert x}^\star+\\ &\qquad\qquad\int_{\Vert Y \Vert_2>K_x}(\Vert Y-\theta \Vert_2-\Vert Y \Vert_2-\langle u,\theta \rangle)\mathrm{d}F_{Y\vert x}^\star\\
&\geq \Vert \theta \Vert_2(\frac{1}{c}F^\star_{Y\vert x}(\Vert Y \Vert_2 \leq K_x)-(1+\Vert u \Vert_2)F^\star_{Y\vert x}(\Vert Y \Vert_2>K_x )\big)\\
&=\Vert \theta \Vert_2\big(\frac{1}{c}-(1+\Vert u \Vert_2+\frac{1}{c})F^\star_{Y\vert x}(\Vert Y \Vert_2>K_x)\big)\\
&\geq \Vert \theta \Vert_2\Big\{\frac{1}{c}-\big(1+\Vert u \Vert_2+\frac{1}{c}\big)\epsilon\Big\}>0.
\end{align*}
Thus, for $u\in B_2^{(k)}$, and every $x\in \mathfrak{X}$, $Q^\star_{Y\vert x}(u) \leq cK_x$, where $K_x$ is chosen such that $P_{Y\vert x}^\star(\Vert Y \Vert_2 \leq K_x)>1-\epsilon$, where $0<\epsilon < c^{-1}/(c^{-1}+\Vert u \Vert_2+1)$. Since the $\delta_n$-smoothed conditional distribution $F_{\delta_n;Y\vert x}$ is weakly consistent at $F^\star_{Y\vert x}$, the posterior probability $Q_{Y\vert x}(u) \leq cK_x$ tends to 1 almost surely.
\end{proof}
\afterpage{\clearpage}
\bibliographystyle{ba}
\bibliography{sample}

\begin{thebibliography}{29}
\providecommand{\enquote}[1]{``#1''}
\expandafter\ifx\csname natexlab\endcsname\relax\def\natexlab#1{#1}\fi
\expandafter\ifx\csname url\endcsname\relax
  \def\url#1{{\tt #1}}\fi
\expandafter\ifx\csname urlprefix\endcsname\relax\def\urlprefix{URL }\fi
\ifx\endbibitem\undefined \let\endbibitem\relax\fi

\bibitem[{Bai et~al.(1990)Bai, Chen, Miao, and
  Radhakrishna~Rao}]{bai1990asymptotic}
Bai, Z., Chen, X., Miao, B., and Radhakrishna~Rao, C. (1990).
\newblock \enquote{Asymptotic theory of least distances estimate in
  multivariate linear models.}
\newblock {\em Statistics\/}, 21(4): 503--519.
\endbibitem

\bibitem[{Bickel and Millar(1992)}]{bickel1992uniform}
Bickel, P. and Millar, P. (1992).
\newblock \enquote{Uniform convergence of probability measures on classes of
  functions.}
\newblock {\em Statistica Sinica\/}, 1--15.
\endbibitem

\bibitem[{Chakraborty(1999)}]{chakraborty1999multivariate}
Chakraborty, B. (1999).
\newblock \enquote{On multivariate median regression.}
\newblock {\em Bernoulli\/}, 5(4): 683--703.
\endbibitem

\bibitem[{Chakraborty(2003)}]{chakraborty2003multivariate}
--- (2003).
\newblock \enquote{On multivariate quantile regression.}
\newblock {\em Journal of Statistical Planning and Inference\/}, 110(1-2):
  109--132.
\endbibitem

\bibitem[{Chang(2015)}]{chang2015nonparametric}
Chang, C. (2015).
\newblock \enquote{Nonparametric Bayesian quantile regression via Dirichlet
  process mixture models.}
\endbibitem

\bibitem[{Chaudhuri(1996)}]{chaudhuri1996geometric}
Chaudhuri, P. (1996).
\newblock \enquote{On a geometric notion of quantiles for multivariate data.}
\newblock {\em Journal of the American Statistical Association\/}, 91(434):
  862--872.
\endbibitem

\bibitem[{Cifarelli and Regazzini(1978)}]{cifarelli1978nonparametric}
Cifarelli, D. and Regazzini, E. (1978).
\newblock \enquote{Nonparametric statistical problems under partial
  exchangeability: The role of associative means.}
\newblock Technical report, Tech. rep., Quaderni Istituto Matematica
  Finanziaria of the University of Turin.
\endbibitem

\bibitem[{De~Iorio et~al.(2004)De~Iorio, M{\"u}ller, Rosner, and
  MacEachern}]{de2004anova}
De~Iorio, M., M{\"u}ller, P., Rosner, G.~L., and MacEachern, S.~N. (2004).
\newblock \enquote{An ANOVA model for dependent random measures.}
\newblock {\em Journal of the American Statistical Association\/}, 99(465):
  205--215.
\endbibitem

\bibitem[{Drovandi and Pettitt(2011)}]{drovandi2011likelihood}
Drovandi, C.~C. and Pettitt, A.~N. (2011).
\newblock \enquote{Likelihood-free Bayesian estimation of multivariate quantile
  distributions.}
\newblock {\em Computational Statistics \& Data Analysis\/}, 55(9): 2541--2556.
\endbibitem

\bibitem[{Duan et~al.(2007)Duan, Guindani, and Gelfand}]{duan2007generalized}
Duan, J.~A., Guindani, M., and Gelfand, A.~E. (2007).
\newblock \enquote{Generalized spatial Dirichlet process models.}
\newblock {\em Biometrika\/}, 94(4): 809--825.
\endbibitem

\bibitem[{Gelfand et~al.(2005)Gelfand, Kottas, and
  MacEachern}]{gelfand2005bayesian}
Gelfand, A.~E., Kottas, A., and MacEachern, S.~N. (2005).
\newblock \enquote{Bayesian nonparametric spatial modeling with Dirichlet
  process mixing.}
\newblock {\em Journal of the American Statistical Association\/}, 100(471):
  1021--1035.
\endbibitem

\bibitem[{Ghosal et~al.(1999)Ghosal, Ghosh, Ramamoorthi
  et~al.}]{ghosal1999posterior}
Ghosal, S., Ghosh, J.~K., Ramamoorthi, R., et~al. (1999).
\newblock \enquote{Posterior consistency of Dirichlet mixtures in density
  estimation.}
\newblock {\em Ann. Statist\/}, 27(1): 143--158.
\endbibitem

\bibitem[{Ghosal and Van~der Vaart(2017)}]{ghosal2017fundamentals}
Ghosal, S. and Van~der Vaart, A. (2017).
\newblock {\em Fundamentals of Nonparametric Bayesian Inference\/}~44.
\newblock Cambridge University Press.
\endbibitem

\bibitem[{Guggisberg(2019)}]{guggisberg2019bayesian}
Guggisberg, M. (2019).
\newblock \enquote{A Bayesian approach to multiple-output quantile regression.}
\newblock {\em arXiv preprint arXiv:1909.02623\/}.
\endbibitem

\bibitem[{Hallin et~al.(2010)Hallin, Paindaveine, {\v{S}}iman, Wei, Serfling,
  Zuo, Kong, and Mizera}]{hallin2010multivariate}
Hallin, M., Paindaveine, D., {\v{S}}iman, M., Wei, Y., Serfling, R., Zuo, Y.,
  Kong, L., and Mizera, I. (2010).
\newblock \enquote{Multivariate quantiles and multiple-output Regression
  quantiles: From $\ell_1$- optimization to halfspace depth [with Discussion
  and Rejoinder].}
\newblock {\em The Annals of Statistics\/}, 635--703.
\endbibitem

\bibitem[{Ifantis and Siafarikas(1990)}]{ifantis1990inequalities}
Ifantis, E. and Siafarikas, P. (1990).
\newblock \enquote{Inequalities involving Bessel and modified Bessel
  functions.}
\newblock {\em Journal of Mathematical Analysis and Applications\/}, 147(1):
  214--227.
\endbibitem

\bibitem[{Ishwaran and James(2001)}]{ishwaran2001gibbs}
Ishwaran, H. and James, L.~F. (2001).
\newblock \enquote{Gibbs sampling methods for stick-breaking priors.}
\newblock {\em Journal of the American Statistical Association\/}, 96(453):
  161--173.
\endbibitem

\bibitem[{Koenker and Bassett~Jr(1978)}]{koenker1978regression}
Koenker, R. and Bassett~Jr, G. (1978).
\newblock \enquote{Regression quantiles.}
\newblock {\em Econometrica: Journal of the Econometric Society\/}, 33--50.
\endbibitem

\bibitem[{Kottas and Gelfand(2001)}]{kottas2001bayesian}
Kottas, A. and Gelfand, A.~E. (2001).
\newblock \enquote{Bayesian semiparametric median regression modeling.}
\newblock {\em Journal of the American Statistical Association\/}, 96(456):
  1458--1468.
\endbibitem

\bibitem[{MacEachern(1999)}]{maceachern1999dependent}
MacEachern, S.~N. (1999).
\newblock \enquote{Dependent nonparametric processes.}
\newblock In {\em ASA Proceedings of the section on Bayesian Statistical
  Science\/}~1, 50--55. Alexandria, Virginia. Virginia: American Statistical
  Association; 1999.
\endbibitem

\bibitem[{Nieto-Barajas et~al.(2012)Nieto-Barajas, M{\"u}ller, Ji, Lu, and
  Mills}]{nieto2012time}
Nieto-Barajas, L.~E., M{\"u}ller, P., Ji, Y., Lu, Y., and Mills, G.~B. (2012).
\newblock \enquote{A time-series DDP for functional proteomics profiles.}
\newblock {\em Biometrics\/}, 68(3): 859--868.
\endbibitem

\bibitem[{Serfling(2002)}]{serfling2002quantile}
Serfling, R. (2002).
\newblock \enquote{Quantile functions for multivariate analysis: approaches and
  applications.}
\newblock {\em Statistica Neerlandica\/}, 56(2): 214--232.
\endbibitem

\bibitem[{Small(1990)}]{small1990survey}
Small, C.~G. (1990).
\newblock \enquote{A survey of multidimensional medians.}
\newblock {\em International Statistical Review/Revue Internationale de
  Statistique\/}, 263--277.
\endbibitem

\bibitem[{Sun et~al.(2017)Sun, Paisley, and Liu}]{sun2017location}
Sun, S., Paisley, J., and Liu, Q. (2017).
\newblock \enquote{Location Dependent Dirichlet Processes.}
\newblock In {\em International Conference on Intelligent Science and Big Data
  Engineering\/}, 64--76. Springer.
\endbibitem

\bibitem[{Tomlinson and Escobar(1999)}]{tomlinson1999analysis}
Tomlinson, G. and Escobar, M. (1999).
\newblock {\em Analysis of densities.\/}.
\newblock University of Toronto Technical report.
\endbibitem

\bibitem[{Van~der Vaart(2000)}]{van2000asymptotic}
Van~der Vaart, A.~W. (2000).
\newblock {\em Asymptotic Statistics\/}~3.
\newblock Cambridge University Press.
\endbibitem

\bibitem[{Waldmann and Kneib(2015)}]{waldmann2015bayesian}
Waldmann, E. and Kneib, T. (2015).
\newblock \enquote{Bayesian bivariate quantile regression.}
\newblock {\em Statistical Modelling\/}, 15(4): 326--344.
\endbibitem

\bibitem[{Wong(1998)}]{wong1998generalized}
Wong, T.-T. (1998).
\newblock \enquote{Generalized Dirichlet distribution in Bayesian analysis.}
\newblock {\em Applied Mathematics and Computation\/}, 97(2-3): 165--181.
\endbibitem

\bibitem[{Yu and Moyeed(2001)}]{yu2001bayesian}
Yu, K. and Moyeed, R.~A. (2001).
\newblock \enquote{Bayesian quantile regression.}
\newblock {\em Statistics \& Probability Letters\/}, 54(4): 437--447.
\endbibitem

\end{thebibliography}


\end{document}